\documentclass[11pt,a4paper]{article}

\usepackage[T1]{fontenc}
\usepackage[utf8]{inputenc}
\usepackage[british]{babel}

\usepackage{settings_tmp}
\usepackage{settings_figures}

\usepackage[round]{natbib}
\usepackage[hidelinks]{hyperref}
\hypersetup{
  pdftitle={Graph determination problem from knowledge of its subgraphs},
  pdfauthor={Coralie Zens et al.}
}
\title{Graph determination problem from knowledge of its subgraphs}

\author[1]{Coralie Zens}
\author[1]{Dominique Barth}
\author[2,3]{Dimitri Watel}
\author[4]{Marc-Antoine Weisser}

\affil[1]{DAVID, Universit{\'e} de Versailles-St-Quentin, France}
\affil[2]{ensIIE, Evry-Courcouronnes, France}
\affil[3]{C{\'e}dric, CNAM, France}
\affil[4]{LISN, CentraleSup{\'e}lec, Universit{\'e} Paris-Saclay, France}
\date{}

\begin{document}

\maketitle

\begin{abstract}
    
    In the context of graph reconstruction studies, this paper addresses 
    the problem defined as follows: given a sample of vertex-coloured graphs, 
    each associated with a number, does there exist a connected graph 
    containing exactly the given number of induced subgraphs isomorphic 
    to each graph in the sample, such that every edge of the solution graph 
    is included in at least one of these subgraphs. This paper first shows 
    that this problem is $\Sigma_2^P$-complete. The associated optimisation 
    problem is to determine the minimum number of vertices in a solution graph, 
    if one exists. 
    For instances consisting of two paths, we prove NP-completeness in general 
    and give a polynomial-time dynamic programming algorithm
    when one path is monochromatic.
    
\end{abstract}


\section{Introduction}\label{sec:intro}

We study the following graph-realisation problem.
The input consists of a collection 
\begin{math}\mathcal{F}=\{G_1,\ldots,G_{|\mathcal{F}|}\}\end{math}
of pairwise non-isomorphic vertex-coloured graphs, an integer $P_i$ for each graph $G_i \in \mathcal{F}$, and an integer $\alpha$.
The question is whether there exists a connected graph $H$ on at most $\alpha$ vertices that contains exactly $P_i$ induced subgraphs isomorphic to $G_i$, for every $i$, and in which every edge belongs to at least one of these prescribed copies.
We call this problem \emph{Supergraph-Sample-Coverage} (SGSC).
Its optimisation version asks for a solution graph of minimum order.
We refer to \cite{Berge73} for standard graph-theoretic terminology.

One motivation arises from chemo-informatics, where molecular structures are represented by graphs and structural information is inferred from physical observables such as vibrational spectra~\citep{Bougueroua25,Konkoli1998I,Konkoli1998III}.

Local structural patterns are often used to characterise molecular configurations and to interpret spectral features. 
This motivates combinatorial abstractions in which such patterns are represented by small graphs and their copies within a larger structure are explicitly constrained.
From this perspective, one may view the inverse problem of predicting a molecular structure from spectral data as the task of constructing a graph that satisfies prescribed constraints on the copies of a collection of patterns, in line with recent formulations of inverse problems in molecular structure elucidation from spectra~\citep{Lemm2024}.
In our formulation, the graphs $G_i$ represent such patterns and the integers $P_i$ encode prescribed multiplicities of these patterns.
The edge-covering condition requires every relation in the reconstructed graph to be supported by an observed pattern.
This yields a combinatorial abstraction inspired by structure-from-spectrum problems, linking spectral considerations with constructive questions in graph theory.

SGSC is related to several classical graph problems, but differs from each of them in a significant way.
In the Graph Reconstruction problem, a graph $G=(V,E)$ is to be uniquely determined from its deck, namely the multiset\footnote{Unlike a set, a multiset allows repeated elements.} of the $|V|$ induced subgraphs $G[V\setminus v]$, one for each $v\in V$~\citep{Ulam60collection,Kelly57congruence}.
The Reconstruction Conjecture, usually attributed to \citeauthor{Kelly57congruence} and \citeauthor{Ulam60collection}, states that every graph with at least three vertices is determined up to isomorphism by its deck, and remains open in general~\citep{Hosaka2022}.
While reconstruction asks whether the original graph is uniquely determined by its vertex-deleted subgraphs, SGSC asks whether there exists any graph satisfying prescribed counts of induced subgraphs.
Introduced by~\cite{Kratsch94gr}, the Legitimate-Deck problem asks whether a given multiset is the deck of some graph.
Like SGSC, it is an existence problem based on induced-subgraph data.
However, its input is a multiset of explicitly given graphs, whereas SGSC specifies the required multiplicities of graph isomorphism types.
Consequently, translating an instance of Legitimate-Deck into an SGSC instance requires accounting for graph isomorphism among the input graphs, so the correspondence is not a direct polynomial-time reduction.

Another related notion is that of an induced-universal graph, which contains every graph in a given family as an induced subgraph~\citep{Chung90}.
Likewise, induced-subgraph isomorphism~\citep{Garey79si} asks whether a single prescribed graph occurs as an induced subgraph, and contain-or-exclude problems~\citep{Ko91three} require certain patterns to occur while others are excluded.
Detecting a prescribed induced subgraph already encompasses classical $NP$-complete problems~\citep{Garey79si}.
Moreover, \citeauthor{Ko91three} showed that related contain-or-exclude problems can be $\Sigma_2^P$-complete; see also~\cite{Stockmeyer76,compendium02} for background on the polynomial hierarchy.
All these notions impose essentially qualitative constraints on induced subgraphs: they specify which patterns must or must not occur, but not how many times they occur.
In contrast, SGSC imposes exact multiplicity constraints.
In particular, accommodating one pattern may inevitably create additional copies of another pattern, which SGSC must explicitly rule out.
These additional counting constraints lead to fundamentally different combinatorial and algorithmic questions.
However, contain-or-exclude problems share a similarity with the exact-count condition in SGSC, which combines the existence of the required copies with the absence of any additional copies.
Finally, instances in which a prescribed multiplicity is equal to one also connect SGSC with the complexity of unique-solution problems~\citep{papadimitriou1984complexity,valiant1986np,Marx08}.

Our contributions are threefold.
First, we prove that SGSC is $\Sigma_2^P$-complete even when exactly three colours are used.
Second, we show that SGSC is $NP$-hard when the sample consists of two paths.
Third, we study a tractable subclass of these two-path instances, where one of the paths is monochromatic.
We characterise the feasible instances and give a polynomial-time dynamic programming algorithm that constructs a minimum-order solution graph.

The remainder of the paper is organised as follows.
Section~\ref{sec:model} introduces the notation and formally defines SGSC.
Section~\ref{sec:comp} establishes its computational complexity.
Finally, Section~\ref{sec:2chain} analyses the two-path tractable instances and presents the dynamic programming algorithm.

\section{Preliminaries}\label{sec:model} 

Before formally defining the problem studied in this paper, we provide key definitions. 

A coloured graph is defined as a tuple $(V, E,\sigma)$ where $V$ is the vertex set, $E$ is the edge set, and \begin{math}\sigma: V\to \mathbb{N}\end{math} is a vertex-colouring function.
This colouring is not necessarily proper, i.e., adjacent vertices may share the same colour.
More generally, a graph, subgraph, set of subgraphs, or set of vertices is said to be of colour $c\in\mathbb{N}$, or coloured in $c$, if all vertices it contains have colour $c$.
For the remainder of this paper, we assume all graphs under consideration are coloured. 

Given two graphs $G= (V,E,\sigma)$ and $G'= (V',E',\sigma')$, an \textbf{isomorphism} $f$ between $G$ and $G'$ is a bijective function from $V$ to $V'$ such that for every $u\in V$, \begin{math}\sigma'(f(u)) = \sigma(u)\end{math}, and for every $u,v \in V$, $(u, v) \in E$ if and only if $(f(u), f(v)) \in E'$.
When such a function exists, graphs $G$ and $G'$ are isomorphic, denoted $G\approx G'$.

Given two graphs $G= (V,E,\sigma)$ and $G'= (V',E',\sigma')$, the graph $G'$ is \textbf{included} in $G$ if and only if $G'$ is isomorphic to an induced subgraph of $G$.

\begin{definition}\label{def:ech}
    A \textbf{sample} is a set of connected graphs $\mathcal{F}=\{G_1, \ldots, G_{|\mathcal{F}|}\}$ such that for $i,j$ with $1\leq i, j\leq |\mathcal{F}|$ and $i\neq j$, $G_i$ is not isomorphic to $G_j$ (see Figure \ref{fig:exg}). 
    
    A sample $\mathcal{F}$ is called \textbf{without inclusion} if and only if, for $i,j$ with $1\leq i, j\leq |\mathcal{F}|$ and $i\neq j$, $G_i$ is not included in $G_j$. 
    Conversely, if it does not satisfy this condition, sample $\mathcal{F}$ is \textbf{with inclusion}.
\end{definition}

Figure \ref{fig:exg} illustrates the concepts of Definition \ref{def:ech}.

\begin{figure}[t]
    \centering
    \hfill
    \begin{subfigure}{0.23\textwidth}
        \centering
        \begin{adjustbox}{max width=0.7\linewidth, max height=0.08\textheight}
            \begin{tikzpicture}[node distance=1.5cm]
                \node[circleNode] (x1) {};
                \node[triangleNode, right of=x1, right=0cm] (x2) {~~};
                \node[triangleNode, above of=x1, right=0.6cm] (x3) {~~};
                
                \draw[-] (x1) -- (x2);
                \draw[-] (x1) -- (x3);
                \draw[-] (x2) -- (x3);
            \end{tikzpicture}
        \end{adjustbox}
        \caption{Graph $G_1$.}     
    \end{subfigure}\hfill
    \begin{subfigure}{0.13\textwidth}
        \centering
        \begin{adjustbox}{max width=\linewidth, max height=0.08\textheight}
            \begin{tikzpicture}[node distance=2cm]
                \node[draw, rectangle, fill=red!20, minimum size=.8cm] (x1) {};
                \node[triangleNode, below of=x1] (x2) {~~};
                
                \draw[-] (x1) -- (x2);
            \end{tikzpicture}
        \end{adjustbox}
        \caption{Graph $G_2$.} 
    \end{subfigure}\hfill
    \begin{subfigure}{0.3\textwidth}
        \centering
        \begin{adjustbox}{max width=\linewidth, max height=0.08\textheight}
            \begin{tikzpicture}[node distance=1.5cm]
                \node[circleNode] (x1) {};
                \node[triangleNode, left of=x1] (x2) {~~};
                \node[triangleNode, right of=x1] (x3) {~~};
                \node[diamondNode, below of=x1, below=-0.7cm] (x4) {~~\quad~};
                
                \draw[-] (x1) -- (x2);
                \draw[-] (x1) -- (x3);
                \draw[-] (x1) -- (x4);
            \end{tikzpicture}
        \end{adjustbox}
        \caption{Graph $G_3$.}
    \end{subfigure}\hfill
    \begin{subfigure}{0.23\textwidth}
        \centering
        \begin{adjustbox}{max width=\linewidth, max height=0.08\textheight}
            \begin{tikzpicture}[node distance=1.5cm]
                \node[draw, rectangle, fill=red!20, minimum size=.8cm] (x1) {};
                \node[triangleNode, right of=x1, right=0cm] (x2) {~~};
                \node[triangleNode, above of=x1, above=0cm] (x3) {~~};
                \node[circleNode, right of=x3] (x4) {};
                
                \draw[-] (x1) -- (x2);
                \draw[-] (x1) -- (x3);
                \draw[-] (x2) -- (x3);
                \draw[-] (x2) -- (x4);
                \draw[-] (x3) -- (x4);
            \end{tikzpicture}
        \end{adjustbox}
        \caption{Graph $G_4$.}
    \end{subfigure}\hfill
    \caption{The sample $\mathcal{F}=\{G_1,G_2,G_3,G_4\}$, a set of coloured graphs. Each form represents a colour. The sample is \emph{with inclusion} as $G_1$ is isomorphic to one induced subgraph of $G_4$, and $G_2$ is isomorphic to two of them.}
    \label{fig:exg}
\end{figure}
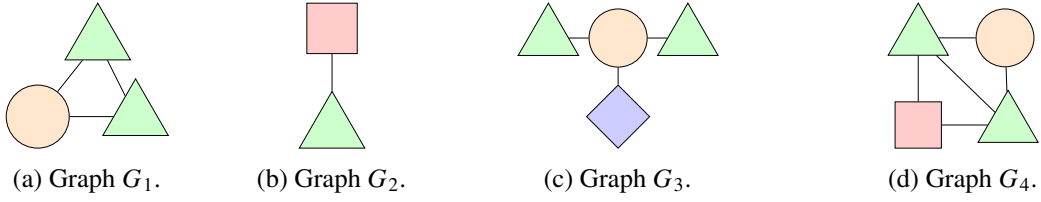

\begin{definition}\label{def:emp}
    Given a sample $\mathcal{F}$, a \textbf{fingerprint} of $\mathcal{F}$ is a vector $P$ of positive integers of dimension $|\mathcal{F}|$. 
    
    The $i^{th}$ component $P_i$ of the vector $P$ is the \textbf{number of copies} of the graph $G_i$. 
    A fingerprint is \textbf{unary} if and only if $P_i=1$ for all $1\leq i\leq |\mathcal{F}|$. 
\end{definition}

Given a graph $G=(V, E,\sigma)$, the multiset of induced subgraphs of $G$ is denoted as $I^G$, and the multiset of induced subgraphs of $G$ with $1\leq k \leq |V|$ vertices is denoted as $I^G_k$. 
Given a graph $H$ and a sample $\mathcal{F}$, the \textbf{cover $C$} of $H$ by $\mathcal{F}$ is a subset of $I^H$ defined by 
\begin{displaymath}
    C = \{h\in I^H: \exists 1\leq i\leq |\mathcal{F}| \text{ such that }h\approx G_i\in \mathcal{F} \}.
\end{displaymath}

Let $H$ be a graph, $\mathcal{F}$ be a sample, and $C$ be the cover of $H$ by $\mathcal{F}$.
An edge $e \in E(H)$ is \textbf{covered by $\mathcal{F}$} if it belongs to an induced subgraph $h \in C$.
A graph $H$ is \textbf{edge-covered by $\mathcal{F}$} if every edge of $H$ is covered by $\mathcal{F}$.

Given a graph $H$ and a sample \begin{math}\mathcal{F}=\{G_1, \ldots, G_{|\mathcal{F}|}\}\end{math}, the \textbf{signature} of $\mathcal{F}$ in $H$ is the vector 
\begin{displaymath}
    S^H = (S^H_1, \ldots, S^H_{|\mathcal{F}|})\qquad\text{with}\qquad
    S^H_i =\left|\{h\in I^H: h \approx G_i \}\right|.
\end{displaymath}

Given a graph $H$, a sample $\mathcal{F}$, and a fingerprint $P$ of $\mathcal{F}$, the \textbf{discrepancy} $\epsilon$ between $H$ and $P$ is
\begin{displaymath}
    \epsilon(H, P)=\sum\limits_{1\leq i \leq |\mathcal{F}|} \left| {P_i-S^H_i}\right|,
\end{displaymath}
where $S^H$ is the signature of $\mathcal{F}$ in $H$. 

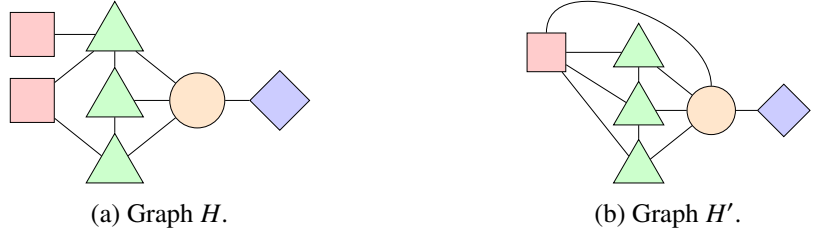
\begin{figure}[ht]
    \centering
    \begin{subfigure}{0.4\textwidth}
        \centering
        \begin{adjustbox}{max width=\linewidth, max height=0.1\textheight}
            \begin{tikzpicture}[node distance=1.5 cm]
                \node[triangleNode] (p1) {~~~} ;
                \node[triangleNode, above of=p1] (p2) {~~~} ;
                \node[triangleNode, below of=p1] (p3) {~~~} ;
                \node[draw, rectangle, fill=red!20, minimum size=.8 cm, left of=p1] (p4) {} ;
                \node[draw, rectangle, fill=red!20, minimum size=.8 cm, left of=p2] (p5) {} ;
                \node[circleNode, right of=p1] (p6) {} ;
                \node[diamondNode, right of=p6] (p7) {~~\quad~} ;
                
                \draw[-] (p1) -- (p2) ;
                \draw[-] (p1) -- (p3) ;
                \draw[-] (p2) -- (p4) ;
                \draw[-] (p3) -- (p4) ;
                \draw[-] (p2) -- (p5) ;
                \draw[-] (p1) -- (p6) ;
                \draw[-] (p2) -- (p6) ;
                \draw[-] (p3) -- (p6) ;
                \draw[-] (p6) -- (p7) ;
            \end{tikzpicture}
        \end{adjustbox}
        \caption{Graph $H$.}
    \end{subfigure}
    \begin{subfigure}{0.45\textwidth}
        \centering
        \begin{adjustbox}{max width=\linewidth, max height=0.12\textheight}
            \begin{tikzpicture}[node distance=1.5 cm]
                \node[triangleNode] (p1) {~~~} ;
                \node[triangleNode, above of=p1] (p2) {~~~} ;
                \node[triangleNode, below of=p1] (p3) {~~~} ;
                \node[draw, rectangle, fill=red!20, minimum size=.8 cm, left of=p2, left=0cm] (p4) {} ;
                \node[circleNode, right of=p1] (p6) {} ;
                \node[diamondNode, right of=p6] (p7) {~~\quad~} ;                
                
                \draw[-] (p1) -- (p2) ;
                \draw[-] (p1) -- (p3) ;
                \draw[-] (p1) -- (p4) ;
                \draw[-] (p2) -- (p4) ;
                \draw[-] (p3) -- (p4) ;
                \draw[-] (p1) -- (p6) ;
                \draw[-] (p2) -- (p6) ;
                \draw[-] (p3) -- (p6) ;
                \draw[-] (p6) -- (p7) ;
                \draw[-] (p4) to[out=90, in=90] (p6) ;
            \end{tikzpicture}
        \end{adjustbox}
        \caption{Graph $H'$.}
    \end{subfigure}
    \caption{Two graphs having signatures of $\mathcal{F}$ from Fig.\ref{fig:exg} equal to the vector $(2,3,1,0)$.}
    \label{fig:exH}
\end{figure}

Finally, given sample $\mathcal{F}$ shown in Figure~\ref{fig:exg}, graphs $H$ and $H'$ in Figure~\ref{fig:exH} have signatures of $\mathcal{F}$ equal to the vector $(2,3,1,0)$. 
The discrepancy $\epsilon$ between $H$ and a fingerprint $P=(1,3,2,0)$ of $\mathcal{F}$ is two, as it is for $H'$. 
What differentiates $H$ and $H'$ is that $H$ is edge-covered by $\mathcal{F}$ while the curved edge in $H'$ is not covered by $\mathcal{F}$.

The decision problem considered is formally defined hereafter. 

\begin{problem}[Supergraph-Sample-Coverage (SGSC)] \label{pb:SGSC}
    \item []\textbf{ Given :} a sample $\mathcal{F}$ of coloured graphs, a fingerprint $P$ and an integer $\alpha>0$.
    \item []\textbf{ Question :}  Does a connected graph $H$ exist, with at most $\alpha$ vertices, such that $\epsilon(H, P)=0$ and is edge-covered by $\mathcal{F}$?
\end{problem}

The optimisation problem associated with this decision problem is the minimisation of $\alpha$, which consists in determining a minimum-order graph satisfying $\epsilon(H, P)=0$ and edge-covered by $\mathcal{F}$. 
This latter condition is referred to as the edge-covering constraint.
In the following, we assume that the values of the fingerprint $P$ are polynomially bounded with respect to the size of the input sample $\mathcal{F}$. This assumption is natural in the considered applications.

For a given instance of SGSC, a connected graph $H$ is said to be a solution graph if $\epsilon(H, P)=0$ and $H$ is edge-covered by $\mathcal{F}$, independently of $\alpha$ and the order of $H$. 
We allow zero entries in the fingerprint, i.e., $P_i=0$ indicates that no induced subgraph isomorphic to $G_i$ may appear in a solution graph.
We define optimal graphs for an instance of SGSC as solution graphs with minimum order. 
For all instances of SGSC, if there is an optimal graph with at most $\alpha$ vertices, then it is a valid solution graph. 
On the other hand, if an optimal graph has more than $\alpha$ vertices, then there is no valid solution graph for this instance. 

For the following section, we assume familiarity with standard complexity classes such as $P$, $NP$, and co-$NP$. 
We also use the second level of the polynomial hierarchy, denoted $\Sigma_2^P$, consisting of decision problems solvable by an $NP$ machine with access to an $NP$ oracle \citep{Stockmeyer76,compendium02}. 
A problem is said to be $\Sigma_2^P$-complete if it belongs to $\Sigma_2^P$ and every problem in $\Sigma_2^P$ can be reduced to it in polynomial time. 
Notably, this class includes the unique-solution SAT problem, which was studied by \cite{papadimitriou1984complexity}, and \cite{valiant1986np}.

\section{The complexity of SGSC}\label{sec:comp}

As discussed in Section \ref{sec:intro}, SGSC differs significantly from related problems. 
Induced-universal graph problems require every graph in a given family to appear as an induced subgraph, whereas SGSC imposes exact counting constraints on induced subgraphs, leading to fundamentally different questions. 
\cite{Kratsch94gr} showed that \textsc{Legitimate-Deck} (LD) is a problem at least as hard as the graph isomorphism problem (GI); however, no precise complexity classification is known for LD. 
Moreover, LD considers a given multiset of subgraphs without resolving graph isomorphism between them, which prevents a direct polynomial-time reduction to SGSC.
These structural and algorithmic differences motivate a dedicated complexity analysis of SGSC. 

We now prove that SGSC is $\Sigma_2^P$-complete. 

\begin{theorem}\label{the:SGSCSigComp}
    SGSC is $\Sigma_2^P$-Complete on graphs with at least $3$ colours.
\end{theorem}

To prove Theorem \ref{the:SGSCSigComp}, we need various intermediate results. 
We first show that SGSC belongs to $\Sigma_2^P$. 
We then introduce an intermediate problem, prove that it is $\Sigma_2^P$-hard via a reduction from $\exists\exists!$ 3-SAT, and finally establish the $\Sigma_2^P$-completeness of SGSC by reducing this intermediate problem to SGSC.
Beforehand, we present two reductions, one from the $k$-Clique problem and one from its complement, to prove that SGSC is $NP$-hard and co-$NP$-hard, respectively. 
These reductions introduce the main ideas used in the proofs of this section.

Given a graph $G$ and an integer $k$, the $k$-Clique problem asks whether there exists a clique of size $k$ in $G$. 
This problem and its complement are, respectively, $NP$-complete and co-$NP$-complete \cite{Garey79si}.

We consider a graph $G$ monochromatic of colour $c\in\mathbb{N}$ and one $(m+1)$-clique $K_{m+1}$ that contained $m$ vertices of colour $c$ and one vertex of colour $c'\in\mathbb{N}$ with $c'\neq c$.
Given an instance of SGSC with a sample $\mathcal{F} = \{G, K_{m+1}\}$, a fingerprint $P=(1,1)$ and $\alpha=|G|+1$, this instance has a valid solution graph if and only if $G$ contains at least one $m$-clique.
Therefore, SGSC is at least as hard as the $k$-Clique problem and is $NP$-hard.

We now consider a graph $G$ monochromatic of colour $c\in\mathbb{N}$ and one $m$-clique $K_m$ with all vertices of colour $c$. 
We define the graph $G'$ as the graph that includes $G$ and $K_m$ plus one vertex $u$ of colour $c'\in\mathbb{N}$ with $c'\neq c$ which is adjacent to two arbitrary vertices in $G$ and $K_m$, respectively.
Given an instance of SGSC with a sample $\mathcal{F} = \{G', K_m\}$, a fingerprint $P=(1,1)$ and $\alpha=|G|+m+1$, this instance has a valid solution graph if and only if $G$ does not contain an $m$-clique.
Therefore, SGSC is at least as hard as the co-$k$-Clique problem and is co-$NP$-hard.

We now return to the proof of Theorem~\ref{the:SGSCSigComp}. 
The reductions presented above introduce several ideas that will later be used in the $\Sigma_2^P$-hardness construction. 
We begin by proving that SGSC belongs to $\Sigma_2^P$.

\begin{lemma}\label{lem:SGSCSig}
    SGSC is contained in $\Sigma_2^P$.
\end{lemma}
\begin{proof}
    Let $\mathcal{F} = \{G_1, \ldots, G_{|\mathcal{F}|}\}$ be a set of graphs, $P$ a vector of $|\mathcal{F}|$ integers and $\alpha$ a positive integer. We consider a non-deterministic algorithm that constructs a graph $H$ of order at most $\alpha$. 

    Given two graphs $G$ and $H$ and an integer $k$, consider the problem of deciding whether $H$ contains $k$ distinct induced subgraphs isomorphic to $G$. This problem is $NP$-complete since it generalises the induced subgraph isomorphism problem. 
    Given access to an oracle solving the previous problem, we show that one can verify the number of copies in $H$ for each graph $G_i$ in $\mathcal{F}$. With an additional verification in polynomial time that every edge is covered, this allows us to verify whether $H$ is a valid solution graph for the given instance.
    
    Counting the copies of a graph $G_i$ amounts to ensuring that $H$ contains at least $k$ induced subgraphs isomorphic to $G_i$ and at most $k$ such subgraphs. 
    We analyse the complexity of these two tasks.
    To confirm that $H$ has at least $k$ induced subgraphs isomorphic to $G_i$, we may directly use the oracle; however, we need to verify that every edge is covered, so we will proceed differently.
    One can guess non-deterministically $k$ distinct vertex subsets, denoted $V_1, \ldots, V_k$, and verify for each subset $V_j$ that the induced subgraph is isomorphic to $G_i$. 
    These $k$ graph isomorphism tests can be verified efficiently using the oracle.
    During these tests, each edge occurring in each subgraph can be labelled as covered. Thus, after dealing with graphs $G_i$ for all $i$, we can check in polynomial time that each edge in $H$ is covered by at least one copy of one such graph $G_i$.
    Conversely, verifying that $H$ contains at most $k$ such subgraphs amounts to checking that no collection of $k+1$ distinct vertex subsets induces subgraphs isomorphic to $G_i$. This corresponds to a co-$NP$ condition with access to the oracle. 
    The case $k=0$ corresponds to verifying that no induced subgraph isomorphic to $G_i$ exists in $H$, which is naturally captured by the co-$NP$ part of the argument.

    Combining these two verifications, one can check, with access to an $NP$ oracle, whether $H$ contains exactly $k$ induced subgraphs isomorphic to $G_i$, while labelling each edge of $H$ covered by one such subgraph. 
    Consequently, one can verify whether $H$ contains exactly $P_i$ induced subgraphs isomorphic to $G_i$ for each $1\leq i\leq|\mathcal{F}|$ and that every edge of $H$ is covered by $\mathcal{F}$ with an $NP$ oracle.
    Since the graph $H$ can be guessed non-deterministically and subsequently verified with access to an $NP$ oracle, SGSC belongs to $\Sigma_2^P$.
\end{proof}

We now introduce the satisfiability problem $\exists\exists!$3-SAT and use the graph problem $\exists\exists!$ Clique as an intermediate problem. 
The reduction below is designed to preserve, for each fixed assignment of the existential variables, the exact number of assignments of the variables subject to the uniqueness quantifier.

Recall that a 3-CNF formula is a conjunction of clauses, each containing exactly three literals.
In what follows, $\uplus$ stands for the union of two non-intersecting sets; the two non-intersecting sets are a partition of this union.

\begin{problem}[$\exists\exists!$ 3-SAT Problem] 
    \item []\textbf{ Given :} a Boolean formula $\varphi$ in 3-CNF form with variables $x_1,\ldots,x_n,y_1,\ldots,y_p$ and $m$ clauses.
    \item []\textbf{ Question :} Does there exist an assignment to $x_1,\ldots,x_n$ such that there exists a unique assignment to $y_1,\ldots,y_p$ for which the formula $\varphi(x_1,\ldots,x_n,y_1,\ldots,y_p)$ is satisfied?
\end{problem}

\begin{problem}[$\exists\exists!$ Clique Problem] 
    \item []\textbf{ Given :} a graph $G = (V_x\uplus V_y,E)$ and two positive integers $k$ and $l$.
    \item []\textbf{ Question :} Does there exist a subset $W_x\subseteq V_x$ of size $k$ such that there exists a unique subset $W_y\subseteq V_y$ of size $l$ for which $W_x\uplus W_y$ induces a $(k+l)$-clique in $G$?
\end{problem}

By Theorem 5 of~\cite{Marx08}, $\exists\exists!$3-SAT is $\Sigma_2^P$-complete. 
We next give a reduction that preserves, for every fixed assignment of the variables in $x$, the number of satisfying assignments of the variables in $y$.

\begin{lemma}\label{lem:exists-unique-clique-hard}
The $\exists\exists!$ Clique problem is $\Sigma_2^P$-hard.
\end{lemma}
\begin{proof}
    Let $\varphi$ be an instance of $\exists\exists!$ 3-SAT with variables $x=(x_1,\ldots,x_n)$ and $y=(y_1,\ldots,y_p)$ and clause set $\mathcal C=\{C_1,\ldots,C_m\}$. 
    We construct an instance $(H(\varphi),V_x,V_y,k,l)$ of $\exists\exists!$ Clique.

    For every variable $x_i$, create two vertices $x_i^0$ and $x_i^1$, and let \begin{math}X_i=\{x_i^0,x_i^1\}.\end{math}
    The vertex $x_i^b$ represents the assignment $x_i=b$. 
    Set $k=n$ and 
    \begin{displaymath}
        V_x=\biguplus_{i=1}^n X_i
    \end{displaymath}

    Similarly, for every variable $y_j$, create two vertices $y_j^0$ and $y_j^1$, let $Y_j=\{y_j^0,y_j^1\}$, and interpret $y_j^b$ as the assignment $y_j=b$.

    For each clause $C\in\mathcal C$, let $\operatorname{Var}(C)$ be the set of distinct variables occurring in $C$. 
    For every assignment
    \begin{displaymath}
        \tau:\operatorname{Var}(C)\longrightarrow\{0,1\}
    \end{displaymath}
    that satisfies $C$, create a vertex $s_{C,\tau}$. 
    Let
    \begin{displaymath}
        S_C=\{s_{C,\tau}:\tau\text{ satisfies }C\}.
    \end{displaymath}
    Since $|\operatorname{Var}(C)|\leq3$, every set $S_C$ has at most eight  vertices. 
    Finally, set $l=p+m$ and 
    \begin{displaymath}
        V_y=
        \left(\biguplus_{j=1}^pY_j\right)
        \uplus
        \left(\biguplus_{C\in\mathcal C}S_C\right).
    \end{displaymath}

    It remains to define the edges of $H(\varphi)$. 
    Each of the sets $X_i$, $Y_j$, and $S_C$ is an independent set. 
    Add all possible edges between vertices representing distinct variables; that is, between $X_i$ and $X_{i'}$ for $i\neq i'$, between $Y_j$ and $Y_{j'}$ for $j\neq j'$, and between every $X_i$ and every $Y_j$. 
    Also add all possible edges between $S_C$ and $S_{C'}$ whenever $C\neq C'$. 
    Finally, join $s_{C,\tau}$ to a variable vertex $z^b$, where $z\in\{x_1,\ldots,x_n,y_1,\ldots,y_p\}$, if and only if either $z\notin\operatorname{Var}(C)$ or $\tau(z)=b$.
    These are all the edges of $H(\varphi)$.

    We prove that the reduction preserves the required quantifiers. 
    First, every clique $W_x\subseteq V_x$ of size $n$ contains exactly one vertex from every independent set $X_i$. 
    Hence, $W_x$ defines a unique complete assignment $a_x$ of the variables in $x$, and every such assignment defines exactly one clique $W_x$ of size $n$.

    Fix such a clique $W_x$ and its assignment $a_x$. We claim that the satisfying assignments $a_y$ of $y$ for $\varphi(a_x,y)$ are in bijection with the subsets $W_y\subseteq V_y$ of size $p+m$ such that $W_x\uplus W_y$ induces a clique.

    Suppose first that $a_y$ satisfies $\varphi(a_x,y)$. 
    Select the vertex $y_j^{a_y(y_j)}$ from each $Y_j$. 
    For every clause $C$, let
    \begin{displaymath}
        \tau_C=(a_x\cup a_y)|_{\operatorname{Var}(C)}.
    \end{displaymath}
    Since $a_x\cup a_y$ satisfies $\varphi$, the assignment $\tau_C$ satisfies $C$, and therefore $s_{C,\tau_C}\in S_C$. 
    Define
    \begin{displaymath}
        W_y(a_y)=
        \{y_j^{a_y(y_j)}:1\leq j\leq p\}
        \cup
        \{s_{C,\tau_C}:C\in\mathcal C\}.
    \end{displaymath}
    This set has size $p+m$. By construction, all its vertices are pairwise adjacent and are adjacent to every vertex of $W_x$. 
    Thus, $W_x\uplus W_y(a_y)$ induces an $(n+p+m)$-clique.

    Conversely, let $W_y\subseteq V_y$ have size $p+m$ and suppose that $W_x\uplus W_y$ induces a clique. The set $V_y$ is partitioned into the $p+m$ independent sets \begin{math}Y_1,\ldots,Y_p,S_{C_1},\ldots,S_{C_m}.\end{math}
    Consequently, $W_y$ contains exactly one vertex from each of these sets.
    Its choices in $Y_1,\ldots,Y_p$ define a complete assignment $a_y$ of the variables in $y$. 
    For every clause $C$, write $s_{C,\tau}$ for the unique vertex of $W_y\cap S_C$. 
    Since this vertex is adjacent to all the variable vertices selected by $W_x\uplus W_y$, the definition of the edges gives
    \begin{displaymath}
    \tau=(a_x\cup a_y)|_{\operatorname{Var}(C)}.
    \end{displaymath}
    But $s_{C,\tau}\in S_C$ only if $\tau$ satisfies $C$. 
    Therefore, $a_x\cup a_y$ satisfies every clause of $\varphi$.

    The two constructions are inverse to each other. 
    Indeed, once $a_x$ and $a_y$ are fixed, the selected vertex in every $Y_j$ is fixed, and the only vertex of $S_C$ adjacent to all selected variable vertices is $s_{C,(a_x\cup a_y)|_{\operatorname{Var}(C)}}$. 
    Hence,
    \begin{displaymath}
        \bigl|\{a_y:\varphi(a_x,a_y)=1\}\bigr|
        =
        \bigl|\{W_y\subseteq V_y:|W_y|=p+m,
                \ H(\varphi)[W_x\uplus W_y]\text{ is a clique}\}\bigr|.
    \end{displaymath}
    It follows, in particular, that there exists exactly one satisfying assignment $a_y$ for the fixed assignment $a_x$ if and only if there exists exactly one corresponding subset $W_y$ for the fixed clique $W_x$. 
    Since assignments of $x$ are themselves in bijection with the cliques $W_x\subseteq V_x$ of size $n$, we obtain
    \begin{displaymath}
        \exists x\;\exists!y\;\varphi(x,y)
        \quad\Longleftrightarrow\quad
        \exists W_x\subseteq V_x\;\exists!W_y\subseteq V_y:
        \begin{array}{l}
            |W_x|=n,\ |W_y|=p+m,\\
            H(\varphi)[W_x\uplus W_y]\text{ is a clique}.
        \end{array}
    \end{displaymath}

    The graph $H(\varphi)$ has at most $2n+2p+8m$ vertices. 
    Its vertex set and edge set can therefore be constructed in polynomial time. 
    Thus, the mapping is a polynomial-time reduction from $\exists\exists!$ 3-SAT to $\exists\exists!$ Clique, and the latter problem is $\Sigma_2^P$-hard.
\end{proof}

We may now use $\exists\exists!$ Clique as the intermediate problem in the reduction to SGSC which proves Theorem \ref{the:SGSCSigComp}.

\begin{proof}{Proof of Theorem \ref{the:SGSCSigComp}}

    Consider an instance of $\exists\exists!$ Clique with a graph $G = (V_x\uplus V_y,E)$ not coloured and two positive integers $k$ and $l$.
    
    We describe a reduction from the problem $\exists\exists!$ Clique to SGSC as follows. 
    We construct a coloured graph $H_1 = (V_x\uplus V_y,E,\sigma_G)$ such that 
    \begin{displaymath}
        \forall v\in V_x, \sigma_G(v)=1\qquad \text{and}\qquad \forall v\in V_y, \sigma_G(v)=2. 
    \end{displaymath}
    Then we construct a coloured graph $H_2$ containing a $(k+l)$-clique of $k$ vertices of colour $1$ and $l$ vertices of colour $2$, and one vertex $u$ of colour $3$ connected to every vertex of colour $1$.
    The construction of $H_1$ and $H_2$ takes $O(|V_x|+|V_y|)$ and $O(k+l)$ times, respectively.
    The instance of SGSC is the sample $\mathcal{F}=\{H_1, H_2\}$, the fingerprint $P=(1,1)$ and $\alpha=|V_x|+|V_y|+1$.
    
    Assume that there exists a subset $W_x\subseteq V_x$ of size $k$ such that there is a unique subset $W_y\subseteq V_y$ of size $l$ for which $W_x\uplus W_y$ induces a $(k+l)$-clique in $G$.
    We construct a graph $H$ including the graph $H_1$, and one vertex $v$ of colour $3$ adjacent to every vertex in $W_x$. 
    
    Since $H_1$ contains no vertex of colour $3$, the subgraph $H_1$ included in $H$ is the unique subgraph of $H$ isomorphic to $H_1$.
    Since all vertices in $V_x$ and $V_y$ are coloured by $\sigma_G$ in colour $1$ and $2$, respectively, and $W_x\uplus W_y$ induces a $(k+l)$-clique in $G$, the subgraph in $H$ induced by $W_x \uplus W_y$ is a $(k+l)$-clique in $H$ with $k$ vertices of colour $1$ and $l$ vertices of colour $2$. 
    Hence, since vertex $v$ of colour $3$ is adjacent to the $k$ vertices of colour $1$ in $W_x$, the subgraph in $H$ induced by $W_x \uplus W_y\uplus \{v\}$ is isomorphic to $H_2$. 
    
    In graph $H$, vertex $v$ is the unique vertex of colour $3$ adjacent to any vertex in $W_x$, so, with the uniqueness of $W_y$, it ensures that there exists a unique subgraph isomorphic to $H_2$. 
    The graph $H$ has exactly $\alpha = |V_x|+|V_y|+1$ vertices. 
    In $H$, every edge with both endpoints in $V_x\uplus V_y$ is covered by $H_1$. 
    In particular, every edge with both endpoints in $W_x\uplus W_y$ is also covered by the unique subgraph isomorphic to $H_2$ in $H$.
    Moreover, every edge in $H$ that does not belong to $H_1$ is incident to $v$ and is covered by the unique subgraph isomorphic to $H_2$ in $H$. 
    Consequently, $H$ is edge-covered by $\mathcal{F}$.
    Hence, if there exists a valid solution graph for the instance of $\exists\exists!$ Clique, then there exists a valid solution graph for the instance of SGSC obtained by the reduction.
    
    Assume that there exists a valid solution graph $H$ for the instance of SGSC obtained by the reduction.
    In graph $H$, there is only one subgraph isomorphic to $H_1$ and only one subgraph isomorphic to $H_2$. 
    Moreover, every edge belongs to at least one of those subgraphs, and we have $|V(H)|\leq|V_x|+|V_y|+1$.
    By definition of $H_1$, this subgraph has $|V_x|$ vertices of colour $1$, $|V_y|$ vertices of colour $2$ and no one of colour $3$.
    By definition of $H_2$, this subgraph has $k$, $l$ and $1$ vertices of colour $1$, $2$ and $3$, respectively.
    Hence, the subgraph of $H$ without the vertex of colour $3$ is isomorphic to $H_1$; the vertices of colour $1$ are in $V_x$, the ones of colour $2$ in $V_y$, and we call $w$ the one of colour $3$.
    The graph $H_2$ contains a $(k+l)$-clique induced by its $k$ vertices of colours $1$ and $l$ vertices of colour $2$, and the $k$ vertices of colour $1$ are adjacent to the vertex of colour $3$.
    Since $H$ contains a subgraph isomorphic to $H_2$, there is a subset $W_x\subseteq V_x$ of size $k$ which is the $k$ vertices adjacent to $w$, and there is a subset $W_y\subseteq V_y$ of size $l$ for which $W_x\uplus W_y\uplus \{w\}$ induces the subgraph isomorphic to $H_2$. 
    Since there is a unique subgraph isomorphic to $H_2$ in $H$, the subset $W_y$ must be the unique subset in $V_y$ with which $W_x$ induces a $(k+l)$-clique in $H$.
    Because the graph $H_1$ is actually the graph $G$ with a colouring, the subset of vertices $W_x\uplus W_y$ induces a $(k+l)$-clique in $G$, and there exists a valid solution graph for the instance of $\exists\exists!$ Clique Problem.
    
    By Lemma \ref{lem:SGSCSig} and this reduction, we conclude that SGSC is $\Sigma_2^P$-complete.
\end{proof}

The $\Sigma_2^P$-completeness of SGSC established above holds in the general setting, where the sample may involve at least three colours. 
This leaves open whether this high complexity is driven by the richness of the sample topology, or persists even for much simpler classes of samples.
We show that the latter holds: when the sample consists of only two prescribed paths, SGSC is $NP$-hard. 
Thus, the complexity of the general problem cannot be attributed solely to the complexity of the sample topology.
We note that we have not established membership in $NP$ for these instances. 
In particular, we are not aware of a straightforward polynomial-size certificate for the corresponding instances.

More precisely, we consider SGSC for samples consisting of two coloured paths, with an arbitrary number of colours, and the fingerprint is unary. 
	
\begin{theorem}
	SGSC is $NP$-hard even if the sample $\mathcal{F}$ contains only two paths and if the fingerprint $P$ is unary. 
\end{theorem}
\begin{proof}
	The reduction is from the \textsc{Maximum Independent Set} problem, which is $NP$-complete for cubic Hamiltonian graphs when the Hamiltonian cycle is given~\citep{Fleischner2010-2742}. 
	Let $G=(V,E)$ be such a graph, let $C=(v_1,\ldots,v_n,v_1)$ be its Hamiltonian cycle, and let $K$ be an integer. 
	We ask whether $G$ contains an independent set of size $K$.
	
	We build an instance of SGSC as follows. 
	The sample $\mathcal{F}$ contains two paths $P_1$ and $P_2$ using $n + 5$ colours $(c_1, c_2, \dots, c_n, c_K, c_{E1}, c_{E2}, c_{P1}, c_{P2})$. 
	Figure~\ref{fig:example} contains an illustration of the two paths.
	
	The path $P_1$ encodes $G$. 
	For every vertex $v_i\in V$, it contains, in this order, the vertices $a_i, k_i, a_i', x_i$ with colours $c_i, c_K, c_i$ and $c_{P1}$ respectively. 
	For every edge $e = (v_i, v_{i+1})\in C$ (where $v_{n + 1} = v_1$), we append $b_i, r_{ei}, s_e, r_{e i+1}$ with colours $c_i, c_{E2}, c_{E1}$ and $c_{E2}$. 
	For every remaining edge $e = (v_i, v_j)$ of $G$, starting with the third incident edge of $v_1$, we append $b_i', r_{ei}, s_e, r_{e j}, b_j', x_e$ with colours $c_i, c_{E2}, c_{E1}, c_{E2}, c_j, c_{P1}$. 
	Since $G$ is cubic, those edges not in $C$ are a matching in $G$ and then $P_1$ contains exactly four vertices of colour $c_i$: $a_i,a_i',b_i$, and $b_i'$.
	
	The path $P_2$ encodes the choice of the independent set and verifies its
	independence. 
	For every $j \in \llbracket 1; K \rrbracket$ the path first contains  $\kappa_j$ and $\psi_j$ with colours $c_K$ and $c_{P2}$. 
	For every $v_i$ of $G$, we then append $\alpha_i, \varphi_i, \beta_i, \phi_i$ respectively with colours $c_i, c_{P2}, c_i, c_{P2}$. 
	Finally, for every edge index $j \in \llbracket 1; |E| \rrbracket$, we add $\gamma_j, \lambda_j, \theta_j$ with colours $c_{E1}, c_{E2}$ and $c_{P2}$.
	
	\begin{figure}[ht!]
		\centering
		\begin{tikzpicture}[]
			\node[label] at (-1.5,1.8) {\Large $P_1$};
			
			\node[ordinary] (a1) at (0,1) {$a_i$\\[-1mm]\scriptsize $c_i$};
			\node[ordinary, right=1mm of a1] (k1) {$k_i$\\[-1mm]\scriptsize $c_K$};
			\node[ordinary, right=1mm of k1] (ap1) {$a_i'$\\[-1mm]\scriptsize $c_i$};
			\node[pathvertex, right=1mm of ap1] (x1) {$x_i$\\[-1mm]\scriptsize $c_{P1}$};
			
			\draw[conn] (a1) -- (k1) -- (ap1) -- (x1);
			
			\node[group, fit=(a1)(k1)(ap1)(x1)] (vertexblock) {};
			\node[label, above=1mm of vertexblock] {for vertices $v_i\in V$};
			
			\node[ordinary, right=12mm of x1] (b1) {$b_i$\\[-1mm]\scriptsize $c_i$};
			\node[edgevertex, right=1mm of b1] (re1) {$r_{e i}$\\[-1mm]\scriptsize $c_{E2}$};
			\node[edgevertex, right=1mm of re1] (se1) {$s_e$\\[-1mm]\scriptsize $c_{E1}$};
			\node[edgevertex, right=1mm of se1] (re2) {$r_{e,i+1}$\\[-1mm]\scriptsize $c_{E2}$};
			
			\draw[conn] (x1) -- (b1) -- (re1) -- (se1) -- (re2);
			
			\node[group, fit=(b1)(re1)(se1)(re2)] (cycleblock) {};
			\node[label, above=1.2mm of cycleblock] (cyclelabel)
			{for edges $e=(v_i,v_{i+1})$ of $C$};
			\node[label, above=-1mm of cycleblock] (cyclelabel)
			{where $v_{n+1} = v_1$};
			
			\node[ordinary, below=12mm of re2] (bp) {$b_i'$\\[-1mm]\scriptsize $c_i$};
			\node[edgevertex, left=1mm of bp] (re3) {$r_{e i}$\\[-1mm]\scriptsize $c_{E2}$};
			\node[edgevertex, left=1mm of re3] (se2) {$s_e$\\[-1mm]\scriptsize $c_{E1}$};
			\node[edgevertex, left=1mm of se2] (re4) {$r_{e j}$\\[-1mm]\scriptsize $c_{E2}$};
			\node[ordinary, left=1mm of re4] (bj) {$b_j'$\\[-1mm]\scriptsize $c_j$};
			\node[pathvertex, left=1mm of bj] (xe) {$x_e$\\[-1mm]\scriptsize $c_{P1}$};
			
			\draw[conn] (re2) -- (bp) -- (re3) -- (se2) --
			(re4) -- (bj) -- (xe);
			
			\node[group, fit=(bp)(re3)(se2)(re4)(bj)(xe)] (otherblock) {};
			\node[label, above=1mm of otherblock] (otherlabel)
			{for other edges $e=(v_i,v_j)$};
			\node[label, below=-2.8mm of otherlabel] 
			{starting with $v_i = v_1$};
			
			\node[label] at (-1.5,-3.5) {\Large $P_2$};
			
			\node[selectedvertex] (kap) at (0,-3.5) {$\kappa_j$\\[-1mm]\scriptsize $c_K$};
			\node[pathvertex, right=1mm of kap] (psi) {$\psi_j$\\[-1mm]\scriptsize $c_{P2}$};
			
			\draw[conn] (kap) -- (psi);
			
			\node[group, fit=(kap)(psi)] (choiceblock) {};
			\node[label, above=1mm of choiceblock]
			{for $j=\llbracket 1; K \rrbracket$};
			
			\node[ordinary, right=12mm of psi] (alpha) {$\alpha_i$\\[-1mm]\scriptsize $c_i$};
			\node[pathvertex, right=1mm of alpha] (phi1) {$\varphi_i$\\[-1mm]\scriptsize $c_{P2}$};
			\node[ordinary, right=1mm of phi1] (beta) {$\beta_i$\\[-1mm]\scriptsize $c_i$};
			\node[pathvertex, right=1mm of beta] (phi2) {$\phi_i$\\[-1mm]\scriptsize $c_{P2}$};
			
			\draw[conn] (psi) -- (alpha) -- (phi1) -- (beta) -- (phi2);
			
			\node[group, fit=(alpha)(phi1)(beta)(phi2)] (vertexcheck) {};
			\node[label, above=1mm of vertexcheck]
			{for vertices $v_i\in V$};
			
			\node[edgevertex, right=12mm of phi2] (gamma) {$\gamma_j$\\[-1mm]\scriptsize $c_{E1}$};
			\node[edgevertex, right=1mm of gamma] (lambda) {$\lambda_j$\\[-1mm]\scriptsize $c_{E2}$};
			\node[pathvertex, right=1mm of lambda] (theta) {$\theta_j$\\[-1mm]\scriptsize $c_{P2}$};
			
			\draw[conn] (phi2) -- (gamma) -- (lambda) -- (theta);
			
			\node[group, fit=(gamma)(lambda)(theta)] (edgecheck) {};
			\node[label, above=1mm of edgecheck]
			{for $j=\llbracket 1; |E| \rrbracket$};
		\end{tikzpicture}
		\caption{Illustration of the two paths in the sample $\mathcal{F}$.}
		\label{fig:example}
	\end{figure}
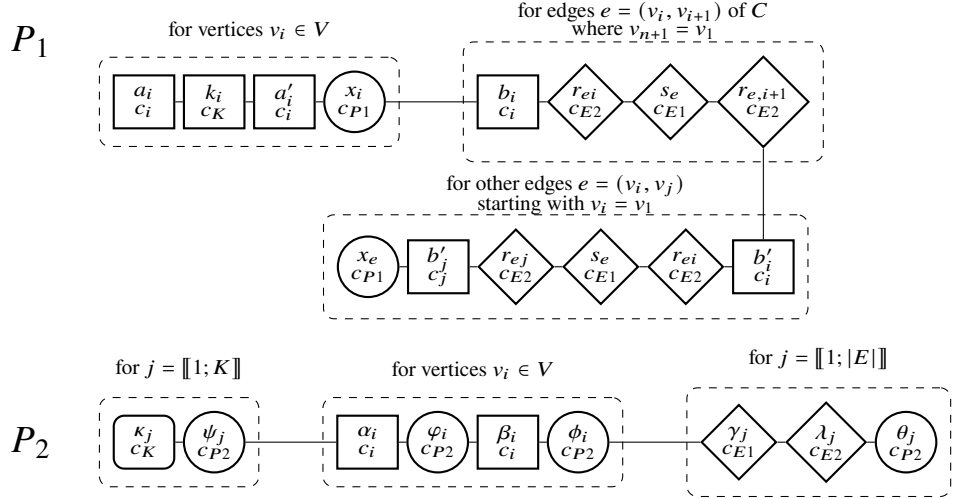
	
	The fingerprint $P$ is unary and the value of $\alpha$ is set to $|P_1| + K + 2 |V| + |E|$. 
	
	Assume first that there exists a graph $H$ with at most $\alpha$ vertices containing unique induced copies of both $P_1$ and $P_2$. 
	The path $P_2$ contains exactly $K+2|V|+|E|$ vertices of colour $c_{P2}$, whereas $P_1$ contains none. 
	Consequently, considering the value of $\alpha$, all remaining vertices of the copy of $P_2$ must be identified with vertices of $P_1$. 
	
	Each vertex $\kappa_j$ must therefore be identified with some vertex $k_i$. 
	These identifications define a set $S$ of $K$ vertices of $G$. 
	We show that $S$ is independent. 
	Given some vertex $v_i\in S$, then $\alpha_i$ and $\beta_i$, coloured with $c_i$, must be identified with two vertices among $a_i, a_i', b_i$ and $b_i'$. 
	Since the copy of $P_2$ contains $k_i$ and is induced, it cannot contain $a_i$ or $a_i'$. 
	Hence, it must contain $b_i$ and $b_i'$.
	
	Now consider an edge $e=(v_i,v_k)$. 
	The vertices $s_e$ of $P_1$ and $\gamma_j$ of $P_2$ are the only vertices of colour $c_{E1}$ in the respective paths, so they must be identified. 
	We assume without loss of generality that $e$ belongs to the Hamiltonian cycle, that the associated part of $P_1$ is $b_i,r_{ei},s_e,r_{ek},b_k$ and that $s_e$ is identified with $\gamma_j$ for some $j$. 
	The vertex $\lambda_j$ must then be identified with either $r_{ei}$ or $r_{ek}$; otherwise the copy of $P_1$ would not be induced. 
	Since the copy of $P_2$ is also induced, it cannot contain simultaneously $r_{ei}$ and $b_i$ or $r_{ek}$ and $b_k$. 
	Therefore, $S$ cannot contain both $v_i$ and $v_k$. 
	As $S$ contains $v_i$ (resp. $v_k$) if and only if $P_2$ contains $b_i$ (resp. $b_k$), $S$ is an independent set of size $K$.

	Conversely, suppose that $S$ is an independent set of size $K$. 
	We construct $H$ by first adding a copy of $P_1$ and then inserting $P_2$. 
	If $v_i$ is the $j$-th vertex of $S$, we identify $k_i$ with $\kappa_j$ and connect them through $\psi_j$. 
	For every vertex $v_i\in V$, if $v_i\in S$, we identify $\alpha_i,\beta_i$ with $b_i,b_i'$. 
	Otherwise, we identify them with $a_i,a_i'$. 
	The corresponding vertices are connected through $\varphi_i$ and $\phi_i$. 
	
	Finally, consider the $j$-th edge $e=(v_i,v_k)$. 
	Since $S$ is independent, at least one endpoint, say $v_i$, does not belong to $S$. 
	We identify $\gamma_j$ with $s_e$ and $\lambda_j$ with $r_{ei}$. 
	We then connect $\lambda_j$ and $\gamma_{j+1}$ through $\theta_j$.
	
	These identifications produce induced copies of both $P_1$ and $P_2$. 
	The resulting graph contains exactly $\alpha$ vertices. 
	Moreover, the copy of $P_1$ is unique because the additional vertices of $P_2$ have colour $c_{P2}$, which does not occur in $P_1$. 
	Every copy of $P_2$ must use all vertices of colour $c_{P2}$ due to the fact that there are as many such vertices in $P_2$ and in $H$, and their degree-one or degree-two structure in $H$ forces the inclusion of their neighbours in the copy. 
	Hence, that copy is also uniquely determined. 
	This completes the reduction. 
\end{proof}

Having established the $NP$-hardness of SGSC for samples consisting of two paths with an arbitrary number of colours, we now ask whether restricting the colour structure can lead to tractable instances.
Specifically, we restrict one of the two paths to be monochromatic. 
Under this additional restriction, we show that the resulting instances can be solved in polynomial time.

\section{Polynomial-Time Solvability with One Monochromatic Path}\label{sec:2chain}
Although SGSC is $NP$-complete for two prescribed paths, some non-trivial instances in this setting remain tractable. 
We establish such a result when one of the two paths is monochromatic.
These instances may use only two colours and form a simple non-trivial setting in which two prescribed patterns interact, since a single-pattern instance can be treated independently. 
The monochromaticity of one path provides enough structure to control their overlap while enforcing exact multiplicities, connectivity, and the global edge-covering constraint.

Our proof has two main components. First, we show that an optimal solution can be represented canonically by a suitable subset of the maximal monochromatic subpaths of $L_1$. Second, we optimise over these subsets by dynamic programming. Thus, the problem reduces to a cardinality-constrained subset-sum computation.

We now introduce some definitions. 
Given a path $L$, an \textbf{extremity} of $L$ is any vertex that serves as an endpoint. 
We consider any path for which its vertices are numbered from one extremity to the other.

Consider two coloured paths $L_1$ and $L_2$ of size $N_1$ and $N_2$, respectively.
We define the path $L_2$ as monochromatic, i.e. every vertex of one colour, which we will call \begin{math}c\in \mathbb{N}\end{math}, and the path $L_1$ contains vertices of various colours, including at least one vertex of colour $c$.

A subpath (i.e., a path contained in a graph) of colour $c$ is \textbf{maximal} if it is not contained in a longer subpath of colour $c$.
The set of all maximal subpaths of colour $c$ in $L_1$ is denoted by $\Lambda$. 
For any \begin{math}\lambda\subseteq\Lambda\end{math}, we call \textbf{end subpath} of $\lambda$ either the first or the last subpath of $\lambda$ according to their order in $L_1$.
Figures~\ref{fig:l1} and~\ref{fig:l2} show an example of the paths $L_1$ and $L_2$, respectively. 

\begin{figure}[ht]
    \centering
    \begin{subfigure}{0.95\textwidth}
        \centering
        \begin{adjustbox}{max width=\linewidth}
            \begin{tikzpicture}[node distance=0.8cm]
                \node[voidNode] (x1) {} ;
                \node[voidNode, right of=x1] (x2) {} ;
                \node[voidNode, right of=x2] (x3) {} ;
                \node[voidNode, right of=x3] (x4) {} ;
                \node[voidNode, right of=x4] (x5) {} ;
                \node[voidNode, right of=x5] (x6) {} ;
                \node[voidNode, right of=x6] (x7) {} ;
                \node[pointNode, right of=x7] (x8) {} ;
                \node[voidNode, right of=x8] (x9) {} ;
                \node[voidNode, right of=x9] (x10) {} ;
                \node[voidNode, right of=x10] (x11) {} ;
                \node[voidNode, right of=x11] (x12) {} ;
                \node[pointNode, right of=x12] (x13) {} ;
                \node[voidNode, right of=x13] (x14) {} ;
                \node[voidNode, right of=x14] (x15) {} ;
                \node[voidNode, right of=x15] (x16) {} ;
                \node[voidNode, right of=x16] (x17) {} ;
                \node[voidNode, right of=x17] (x18) {} ;
                \node[voidNode, right of=x18] (x19) {} ;
                \node[voidNode, right of=x19] (x20) {} ;
                \node[voidNode, right of=x20] (x21) {} ;
                \node[pointNode, right of=x21] (x22) {} ;

                \draw[-] (x1) -- (x2) ;
                \draw[-] (x2) -- (x3) ;
                \draw[-] (x3) -- (x4) ;
                \draw[-] (x4) -- (x5) ;
                \draw[-] (x5) -- (x6) ;
                \draw[-] (x6) -- (x7) ;
                \draw[-] (x7) -- (x8) ;
                \draw[-] (x8) -- (x9) ;
                \draw[-] (x9) -- (x10) ;
                \draw[-] (x10) -- (x11) ;
                \draw[-] (x11) -- (x12) ;
                \draw[-] (x12) -- (x13) ;
                \draw[-] (x13) -- (x14) ;
                \draw[-] (x14) -- (x15) ;
                \draw[-] (x15) -- (x16) ;
                \draw[-] (x16) -- (x17) ;
                \draw[-] (x17) -- (x18) ;
                \draw[-] (x18) -- (x19) ;
                \draw[-] (x19) -- (x20) ;
                \draw[-] (x20) -- (x21) ;
                \draw[-] (x21) -- (x22) ;

                \draw [decorate,decoration={brace,amplitude=.35cm,raise=.06cm}]
                ([yshift=.1cm]x1.north west) -- ([yshift=.1cm]x7.north east) 
                node [black, midway, yshift=18pt] {$l_1$};
                \draw [decorate,decoration={brace,amplitude=.35cm,raise=.06cm}]
                ([yshift=.1cm]x9.north west) -- ([yshift=.1cm]x12.north east) 
                node [black, midway, yshift=18pt] {$l_2$};
                \draw [decorate,decoration={brace,amplitude=.35cm,raise=.06cm}]
                ([yshift=.1cm]x14.north west) -- ([yshift=.1cm]x21.north east) 
                node [black, midway, yshift=18pt] {$l_3$};
            \end{tikzpicture}
        \end{adjustbox}
        \caption{Path $L_1$ with 22 vertices of various colours including $c$.}
        \label{fig:l1}
    \end{subfigure}
    \begin{subfigure}{0.95\textwidth}
        \centering
        \begin{adjustbox}{max width=\linewidth}
            \begin{tikzpicture}[node distance=0.8cm]
                \node (x0) {} ;
                \node[below of=x0,above=.3cm] (x1) {} ;
                \node[right of=x1] (x2) {} ;
                \node[right of=x2] (x3) {} ;
                \node[right of=x3] (x4) {} ;
                \node[right of=x4] (x5) {} ;
                \node[voidNode, right of=x5] (x6) {} ;
                \node[voidNode, right of=x6] (x7) {} ;
                \node[voidNode, right of=x7] (x8) {} ;
                \node[voidNode, right of=x8] (x9) {} ;
                \node[voidNode, right of=x9] (x10) {} ;
                \node[voidNode, right of=x10] (x11) {} ;
                \node[voidNode, right of=x11] (x12) {} ;
                \node[voidNode, right of=x12] (x13) {} ;
                \node[voidNode, right of=x13] (x14) {} ;
                \node[voidNode, right of=x14] (x15) {} ;
                \node[right of=x15] (x16) {} ;
                \node[right of=x16] (x17) {} ;
                \node[right of=x17] (x18) {} ;
                \node[right of=x18] (x19) {} ;
                \node[right of=x19] (x20) {} ;
                \node[right of=x20] (x21) {} ;
                \node[right of=x21] (x22) {} ;

                \draw[-] (x6) -- (x7) ;
                \draw[-] (x7) -- (x8) ;
                \draw[-] (x8) -- (x9) ;
                \draw[-] (x9) -- (x10) ;
                \draw[-] (x10) -- (x11) ;
                \draw[-] (x11) -- (x12) ;
                \draw[-] (x12) -- (x13) ;
                \draw[-] (x13) -- (x14) ;
                \draw[-] (x14) -- (x15) ;
            \end{tikzpicture}
        \end{adjustbox}
        \caption{Path $L_2$ with 10 vertices of colour $c$.}
        \label{fig:l2}
    \end{subfigure}
    \caption{The sample of an instance of SGSC with the path $L_1$ and the path $L_2$. 
    The vertices of colour $c$ are empty circles, and the vertices of other colours are solid circles.
    In Fig.\ref{fig:l1}, the braces show the three maximal subpaths of colour $c$ in the set $\Lambda$.}
    \label{fig:l1l2}
\end{figure}
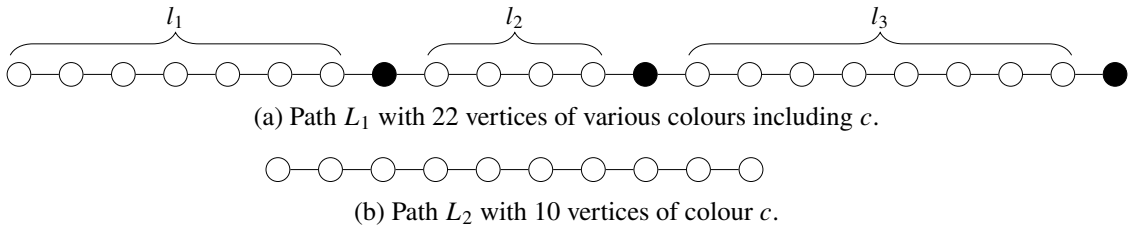

We denote by $\mathcal{I}$ the set of all instances of SGSC that consist of a sample \begin{math}\mathcal{F}=\{L_1, L_2\}\end{math}, a unary fingerprint $P$ (see Definition~\ref{def:emp}) and an integer $\alpha>0$. 
This section studies instances within $\mathcal{I}$, and establishes this result.

\begin{theorem}\label{the:classP}
    SGSC is solvable in polynomial time for instances in $\mathcal{I}$. 
\end{theorem}

We first dispose of the immediate cases in Section~\ref{sec:trivial}. For the remaining instances, Section~\ref{sec:propert} introduces the structural encoding. 
Section~\ref{sec:construction} proves the canonical-solution theorem: restricting the search to graphs generated from perfect subsets of $\Lambda$ does not exclude all optimal solutions. 
Finally, Section~\ref{sec:subpb} finds an optimal perfect subset in polynomial time.

Since an instance is feasible precisely when the order of an optimal solution is at most $\alpha$, we work with the optimisation version throughout the structural and algorithmic parts of the proof.

\subsection{Instances with inclusion or small vertex bound}\label{sec:trivial}

We define $\mathcal{J}$ as the subset of $\mathcal{I}$ containing all instances whose sample is \emph{without inclusion} (see Definition~\ref{def:ech}) and satisfying \begin{math}\alpha>\max\{N_1,N_2\}.\end{math}
We analyse the complexity of SGSC for instances in \begin{math}\mathcal{I}\setminus\mathcal{J}.\end{math}

\begin{lemma}\label{lem:IexJ}
    SGSC is solvable in polynomial time for instances in \begin{math}\mathcal{I}\setminus\mathcal{J}.\end{math}
\end{lemma}
\begin{proof}
    Every instance in \begin{math}\mathcal{I}\setminus\mathcal{J}\end{math} either has a sample \emph{with inclusion} or satisfies \begin{math}\alpha\leq\max\{N_1,N_2\}.\end{math}
    
    First, we consider the instances of $\mathcal{I}$ having a sample \emph{with inclusion}.
    Recall that the sample contains a path $L_1$ with vertices of various colours, including $c$, and a monochromatic path $L_2$ of colour $c$. 
    If the sample is \emph{with inclusion}, then either $L_1$ is included in $L_2$ or $L_2$ is included in $L_1$.
    If $L_1$ is included in $L_2$, then $L_1$ must also be monochromatic, and, since two graphs in a sample are not isomorphic, it satisfies $N_1<N_2$. 
    Consequently, $L_2$ contains several subpaths isomorphic to $L_1$, violating the multiplicity constraints. 
    Hence, no valid solution graph exists.
    If $L_2$ is included in $L_1$, then $N_1>N_2$ and some subpath of $\Lambda$ contains at least $N_2$ vertices. 
    In this case, $L_1$ is a candidate for a solution graph, provided that $\alpha\geq N_1$, since every edge of $L_1$ is covered by its unique occurrence.
    Therefore, a valid solution graph exists, and $L_1$ itself is such graph, if and only if $\alpha\geq N_1$ and the longest maximal subpath of $\Lambda$ is unique and contains exactly $N_2$ vertices.
    These instances can therefore be recognised and solved in polynomial time by scanning $L_1$ and, if $L_1$ is not coloured in $c$, by sizing the subpaths in $\Lambda$.
    
    Secondly, we consider the instances of $\mathcal{I}$ having a sample \emph{without inclusion} and \begin{math}\alpha\leq\max\{N_1,N_2\}.\end{math}
    If the sample is \emph{without inclusion}, then every solution graph must contain distinct occurrences of both $L_1$ and $L_2$. 
    Consequently, every solution graph contains strictly more than $\max\{N_1,N_2\}$ vertices.
    Since \begin{math}\alpha\leq\max\{N_1,N_2\}\end{math}, no valid solution graph exists.
    These instances can therefore be recognised and solved in polynomial time. 
\end{proof}

In the following, we consider the instances of $\mathcal{J}$.
Recall that to solve an instance of SGSC, it suffices to determine the size of an optimal solution graph.
Consequently, the structure of an instance in $\mathcal{J}$ is entirely determined by the triplet $(L_1,N_2,c)$ where $L_1$ is the coloured path and $N_2$ is the number of vertices of colour $c$ in the monochromatic path $L_2$, independently of the value of $\alpha$.
Hereafter, we will denote instances of $\mathcal{J}$ via triplets $(L_1, N_2, c)$.

\subsection{Structural Encoding of Solutions}\label{sec:propert}

Before proceeding further, we provide some definitions.

Given a triplet $(L_1, N_2, c)$ and a solution graph $H$, we will denote the unique induced subpaths of $H$ isomorphic to $L_1$ and $L_2$, by $H_1$ and $H_2$ respectively.
Observe that $H_2$ is the maximum induced subpath of colour $c$ in a solution graph $H$. 

Figure~\ref{fig:sol} provides examples of solution graphs with the paths $L_1$ and $L_2$ shown in Figure~\ref{fig:l1l2}. 

\begin{figure}[t]
    \centering
    \begin{subfigure}{0.9\textwidth}
        \centering
        \begin{adjustbox}{max width=\linewidth}
            \begin{tikzpicture}[node distance=.8cm]
                \node[voidNode+] (x1) {} ;
                \node[voidNode+, right of=x1] (x2) {} ;
                \node[voidNode+, right of=x2] (x3) {} ;
                \node[voidNode+, right of=x3] (x4) {} ;
                \node[voidNode+, right of=x4] (x5) {} ;
                \node[voidNode, right of=x5] (x6) {} ;
                \node[voidNode, right of=x6] (x7) {} ;
                \node[pointNode, right of=x7] (x8) {} ;
                \node[voidNode+, right of=x8] (x9) {} ;
                \node[voidNode+, right of=x9] (x10) {} ;
                \node[voidNode+, right of=x10] (x11) {} ;
                \node[voidNode+, right of=x11] (x12) {} ;
                \node[pointNode, right of=x12] (x13) {} ;
                \node[voidNode, right of=x13] (x14) {} ;
                \node[voidNode, right of=x14] (x15) {} ;
                \node[voidNode, right of=x15] (x16) {} ;
                \node[voidNode, right of=x16] (x17) {} ;
                \node[voidNode, right of=x17] (x18) {} ;
                \node[voidNode, right of=x18] (x19) {} ;
                \node[voidNode, right of=x19] (x20) {} ;
                \node[voidNode, right of=x20] (x21) {} ;
                \node[pointNode, right of=x21] (x22) {} ;
                \node[voidNode+, below of=x7] (j1) {} ;

                \draw[-] (x1) -- (x2) ;
                \draw[-] (x2) -- (x3) ;
                \draw[-] (x3) -- (x4) ;
                \draw[-] (x4) -- (x5) ;
                \draw[-] (x5) -- (x6) ;
                \draw[-] (x6) -- (x7) ;
                \draw[-] (x7) -- (x8) ;
                \draw[-] (x8) -- (x9) ;
                \draw[-] (x9) -- (x10) ;
                \draw[-] (x10) -- (x11) ;
                \draw[-] (x11) -- (x12) ;
                \draw[-] (x12) -- (x13) ;
                \draw[-] (x13) -- (x14) ;
                \draw[-] (x14) -- (x15) ;
                \draw[-] (x15) -- (x16) ;
                \draw[-] (x16) -- (x17) ;
                \draw[-] (x17) -- (x18) ;
                \draw[-] (x18) -- (x19) ;
                \draw[-] (x19) -- (x20) ;
                \draw[-] (x20) -- (x21) ;
                \draw[-] (x21) -- (x22) ;

                \draw[-] (j1) to[out=180, in=-90] (x5) ;
                \draw[-] (j1) to[out=0, in=-90] (x9) ;
            \end{tikzpicture}
        \end{adjustbox}
        \caption{Optimal graph which subpath set is $\{l_1,l_2\}$.}
        \label{fig:sol1}
    \end{subfigure}
    \begin{subfigure}{0.9\textwidth}
        \centering
        \begin{adjustbox}{max width=\linewidth}
            \begin{tikzpicture}[node distance=.8cm]
                \node (x0) {} ;
                \node[voidNode+, below of=x0] (x1) {} ;
                \node[voidNode, right of=x1] (x2) {} ;
                \node[voidNode, right of=x2] (x3) {} ;
                \node[voidNode, right of=x3] (x4) {} ;
                \node[voidNode, right of=x4] (x5) {} ;
                \node[voidNode, right of=x5] (x6) {} ;
                \node[voidNode, right of=x6] (x7) {} ;
                \node[pointNode, right of=x7] (x8) {} ;
                \node[voidNode+, right of=x8] (x9) {} ;
                \node[voidNode+, right of=x9] (x10) {} ;
                \node[voidNode+, right of=x10] (x11) {} ;
                \node[voidNode+, right of=x11] (x12) {} ;
                \node[pointNode, right of=x12] (x13) {} ;
                \node[voidNode, right of=x13] (x14) {} ;
                \node[voidNode, right of=x14] (x15) {} ;
                \node[voidNode, right of=x15] (x16) {} ;
                \node[voidNode+, right of=x16] (x17) {} ;
                \node[voidNode+, right of=x17] (x18) {} ;
                \node[voidNode+, right of=x18] (x19) {} ;
                \node[voidNode+, right of=x19] (x20) {} ;
                \node[voidNode+, right of=x20] (x21) {} ;
                \node[pointNode, right of=x21] (x22) {} ;
                \node[voidNode+, below of=x14] (j1) {} ;
    
                \draw[-] (x1) -- (x2) ;
                \draw[-] (x2) -- (x3) ;
                \draw[-] (x3) -- (x4) ;
                \draw[-] (x4) -- (x5) ;
                \draw[-] (x5) -- (x6) ;
                \draw[-] (x6) -- (x7) ;
                \draw[-] (x7) -- (x8) ;
                \draw[-] (x8) -- (x9) ;
                \draw[-] (x9) -- (x10) ;
                \draw[-] (x10) -- (x11) ;
                \draw[-] (x11) -- (x12) ;
                \draw[-] (x12) -- (x13) ;
                \draw[-] (x13) -- (x14) ;
                \draw[-] (x14) -- (x15) ;
                \draw[-] (x15) -- (x16) ;
                \draw[-] (x16) -- (x17) ;
                \draw[-] (x17) -- (x18) ;
                \draw[-] (x18) -- (x19) ;
                \draw[-] (x19) -- (x20) ;
                \draw[-] (x20) -- (x21) ;
                \draw[-] (x21) -- (x22) ;
    
                \draw[-] (j1) to[out=175, in=-75] (x12) ;
                \draw[-] (j1) to[out=5, in=-125] (x17) ;
            \end{tikzpicture}
        \end{adjustbox}
        \caption{Optimal graph which subpath set is $\{l_2,l_3\}$.}
        \label{fig:sol2}
    \end{subfigure}
    \caption{Optimal graphs for the instance given in Figure \ref{fig:l1l2}. 
    The vertices of colour $c$ are empty circles, and the vertices of other colours are solid circles.
    The vertices in maximal subpaths of colour $c$ are highlighted in grey.
    In each graph, the \emph{added} vertex is below subpath $H_1$.}
    \label{fig:sol}
\end{figure}
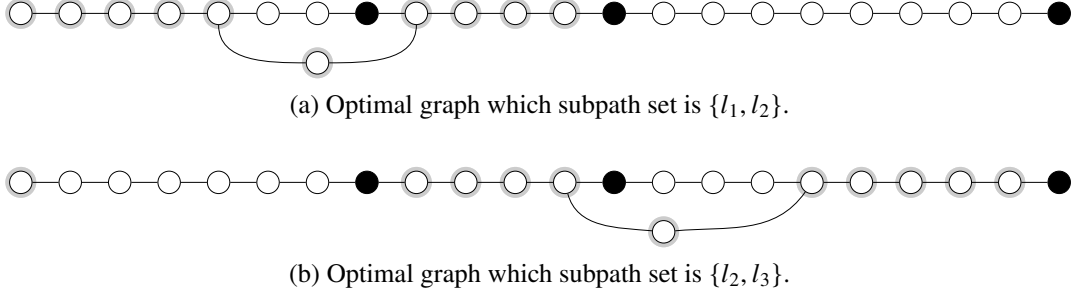

\begin{definition}\label{def:set}
    Given a solution graph, the vertices of colour $c$ are separated into three sets, which are : 
    \begin{itemize}
        \item $V_{1+2}$ is the set of \emph{used} vertices which are in $H_1$ and $H_2$. 
        \item $V_{1-2}$ is the set of \emph{ignored} vertices which are in $H_1$ but not in $H_2$. 
        \item $V_{2-1}$ is the set of \emph{added} vertices which are in $H_2$ but not in $H_1$. 
    \end{itemize}
\end{definition}

Since any edge of a solution graph $H$ must belong to $H_1$ or $H_2$, every vertex in $H$ belongs to either $V_{1-2}$, $V_{2-1}$, or $V_{1+2}$. 

For the following, we need to define some terms related to the subpaths of $\Lambda$ (i.e., the maximal subpaths of colour $c$ in $L_1$) within a solution graph.

\begin{definition}\label{def:lambda}
    Given a solution graph $H$, we define its subpath set $\lambda_H\subseteq\Lambda$ as the set of maximal subpaths in $\Lambda$ containing at least one vertex of $V_{1+2}$.
\end{definition}

We remind that the vertices of path $L_1$ are numbered from one extremity to the other. 
Therefore, we number the subpaths in a set $\lambda\subseteq\Lambda$ following the same order in $L_1$ from $l_1$ to $l_{|\lambda|}$. 

\begin{definition}\label{def:funcf}
    Given a set $\Lambda$, we define a function $f$ such that for all $\lambda\subseteq\Lambda$, 
    \begin{displaymath}
        f(\lambda)=\left(\sum_{l_i \in \lambda} |l_i|\right) + (|\lambda|-1).
    \end{displaymath}
\end{definition}

Given a solution graph $H$ and its subpath set $\lambda_H$, the value $f(\lambda_H)$ depends on the number and sizes of subpaths in $\lambda_H$. If the solution graph $H$ is the result of Algorithm~$\mathcal{A}_1$ (see p.~\pageref{alg:construct}), this value is related to the order of $H$ (see Lemma~\ref{lem:Hsize}). Intuitively, larger values of $f(\lambda_H)$ correspond to constructions with more vertices in $V_{1+2}$.
Provided that $N_2> f(\lambda_H)$, we will see that, if $f(\lambda_H)$ increases, then $|V(H)|$ decreases.

The following notion of a perfect subset identifies subpaths of $\Lambda$ that can jointly participate in the construction of a solution graph. Conditions~\ref{cond:notempty},~\ref{cond:twochains} and~\ref{cond:ign} respectively ensure connectivity, prevent additional occurrence of $L_2$, and exclude unnecessary subpaths.

\begin{definition}\label{def:perfect}
    Consider a value $N_2$ and a subset $\lambda\subseteq\Lambda$.
    We say $\lambda$ is \textbf{perfect} if it meets the following Conditions~\ref{cond:notempty},~\ref{cond:twochains} and~\ref{cond:ign}:
    \begin{enumerate}[(a)]
        \item $|\lambda|>0$ ;\label{cond:notempty}
        \item if $\lambda=\{l_1, l_2\}$ and \begin{math}N_2<f(\lambda)\end{math}, then 
        \begin{math}f(\lambda)-N_2 \leq\left\lfloor\frac{|l_1|}{2}\right\rfloor +\left\lfloor\frac{|l_2|}{2}\right\rfloor -2\end{math} 
        ;\label{cond:twochains}
        \item for all \begin{math}l\in\lambda, |l|>f(\lambda) - N_2\end{math}. \label{cond:ign}
    \end{enumerate}
\end{definition}

In Figure~\ref{fig:l1l2}, the subset $\{l_1,l_3\}\subset\Lambda$ is not perfect because it violates Condition~\ref{cond:twochains}, whereas $\Lambda=\{l_1,l_2,l_3\}$ violates Condition~\ref{cond:ign}.

In the following section, we describe a canonical construction algorithm for an optimal solution graph.

\subsection{Canonical Construction and Normalisation}\label{sec:construction}
The main idea is to restrict the search space to graphs constructed from perfect subsets of $\Lambda$. 
The following construction associates a candidate graph to every perfect subset $\lambda\subseteq\Lambda$. 
We first establish structural properties of the graphs returned by Algorithm~$\mathcal{A}_1$, and then show that minimising a specific function over perfect subsets yields an optimal solution graph.

Given an instance $(L_1, N_2, c)$ of $\mathcal{J}$ and a perfect subset $\lambda\subseteq\Lambda$, Algorithm~$\mathcal{A}_1$ constructs a graph $H$. 

\begin{algorithm}[b!]
    \SetAlgoLined
    \DontPrintSemicolon
    \KwIn{$L_1$ (coloured path), $N_2$ ($L_2$ length), $c$ (colour), $\lambda$ (perfect subset of $\Lambda$)}
    \KwOut{Graph $H$ constructed}
    Initialise graph $H$ to be isomorphic to $L_1$\;\label{alg:line:l1}
    Order paths of $\lambda$ in $L_1$ from $l_1$ to $l_{z}$ with $z=|\lambda|$, following the vertex numbering in $L_1$ \;\label{alg:line:order}
    Create $z-1$ vertices $v_1, v_2, \ldots v_{z-1}$ of colour $c$\;\label{alg:line:v}
    Set \begin{math}x = f(\lambda) =\left(\sum_{l_i \in \lambda} |l_i|\right) + (|\lambda|-1)\end{math} \;\label{alg:line:x}
    \uIf{\begin{math}N_2 \geq x\end{math} \quad \textnormal{\texttt{($H$ may have more added vertices)}}}{
        Create a path $a$ of $N_2 - x$ vertices of colour $c$ that can be empty\;\label{alg:line:a}
        \uIf{$z=1$}{
            Connect the last vertex of path $a$ to an extremity of path $l_1$ coinciding with an extremity of $L_1$; if none exists, connect it to either extremity of $l_1$ \;\label{alg:line:alink1}
        }
        \Else{
            \uIf{$|l_1|\geq|l_z|$ \label{alg:line:test}}{
                Connect the first vertex of $l_1$ to the last vertex of $a$ \;\label{alg:line:alink2}
            }
            \Else{
                Connect the last vertex of $l_z$ to the first vertex of $a$ \;\label{alg:line:alink3}
            }
            Connect $v_1$ to the last vertex of $l_1$ and the first vertex of $l_2$ \;\label{alg:line:y1link}
        }
    }
    \uElseIf{$N_2 < x$ \quad \textnormal{\texttt{($H$ will have ignored vertices)}}}{
        \uIf{$z = 2$}{
            Set \begin{math}y=\min\left\{x-N_2,\left\lfloor\frac{|l_1|}{2}\right\rfloor-1\right\}\end{math} and $y'=x-N_2-y$ \;\label{alg:line:y2}
            Connect $v_1$ to the $(|l_1|-y)$-th vertex of $l_1$ and the $(y'+1)$-th vertex of $l_2$ 
            \tcp*{Since $\lambda$ is perfect, $|l_2|>y'$}\label{alg:line:y2link}
        }
        \uElseIf{$z > 2$}{
            Set $y'= x - N_2$ the number of ignored vertices \;\label{alg:line:y3}
            Connect $v_1$ to the last vertex of $l_1$ and the $(y' +1)$-th vertex of $l_2$ \tcp*{Same comment than Line \ref{alg:line:y2link}}\label{alg:line:y3link}
        }
    }
    \ForEach{$2 \leq i < z$}{
        Connect $v_i$ to the last vertex of $l_i$ and the first vertex of $l_{i+1}$  \; \label{alg:line:chained}
    }
    \Return{$H$}\;
    \caption{$\mathcal{A}_1$ - Graph construction based on SGSC instance in $\mathcal{J}$ and a perfect subset $\lambda$}
    \label{alg:construct}
\end{algorithm}

Figure~\ref{fig:alg-ex} illustrates various graphs returned by Algorithm~$\mathcal{A}_1$ depending on $|\lambda|$ and $f(\lambda)$.
Observe that Algorithm~$\mathcal{A}_1$ constructs a graph in $O(N_1\cdot|\lambda|)$ time.

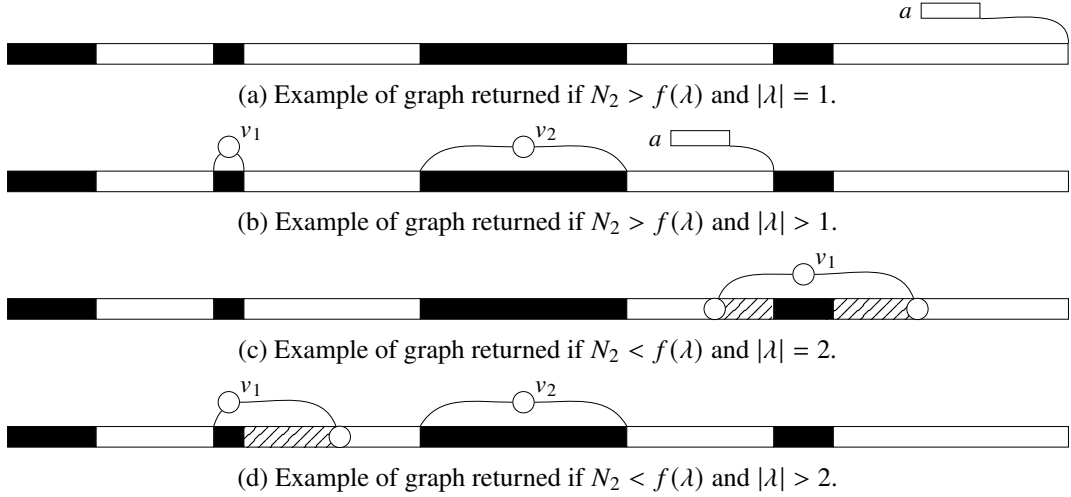
\begin{figure}[t]
    \centering
    \begin{subfigure}{0.9\textwidth}
        \centering
        \begin{adjustbox}{max width=\linewidth}
            \begin{tikzpicture}[node distance=1cm]
                \node[band, fill=black, minimum height=.35cm, minimum width=1.5cm] (l1) {};
                \node[band, fill=white, right of=l1, right=-.25cm, minimum height=.35cm, minimum width=2cm] (l2) {};
                \node[band, fill=black, right of=l2, right=0cm, minimum height=.35cm, minimum width=0.5cm] (l3) {};
                \node[band, fill=white, right of=l3, right=-.75cm, minimum height=.35cm, minimum width=3cm] (l4) {};
                \node[band, fill=black, right of=l4, right=.5cm, minimum height=.35cm, minimum width=3.5cm] (l5) {};
                \node[band, fill=white, right of=l5, right=.75cm, minimum height=.35cm, minimum width=2.5cm] (l6) {};
                \node[band, fill=black, right of=l6, right=.25cm, minimum height=.35cm, minimum width=1cm] (l7) {};
                \node[band, fill=white, right of=l7, right=-.5cm, minimum height=.35cm, minimum width=4cm] (l8) {};
                
                \node[band, fill=white, above of=l8, above=-.4cm, minimum width=1cm] (a) {};
                \draw ([xshift=-.25cm,yshift=.1cm]a.south west) node {\large$a$};
                
                \draw[-] (l8.north east) to[out=90, in=0] (a.south east);
            \end{tikzpicture}
        \end{adjustbox}
        \caption{Example of graph returned if \begin{math}N_2>f(\lambda)\end{math} and $|\lambda|=1$.}
        \label{fig:alg-chainA1}
    \end{subfigure}
    \begin{subfigure}{0.9\textwidth}
        \centering
        \begin{adjustbox}{max width=\linewidth}
            \begin{tikzpicture}[node distance=1cm]
                \node[] (x) {} ;
                \node[band, fill=black, below of=x, minimum height=.35cm, minimum width=1.5cm] (l1) {};
                \node[band, fill=white, right of=l1, right=-.25cm, minimum height=.35cm, minimum width=2cm] (l2) {};
                \node[band, fill=black, right of=l2, right=0cm, minimum height=.35cm, minimum width=0.5cm] (l3) {};
                \node[band, fill=white, right of=l3, right=-.75cm, minimum height=.35cm, minimum width=3cm] (l4) {};
                \node[band, fill=black, right of=l4, right=.5cm, minimum height=.35cm, minimum width=3.5cm] (l5) {};
                \node[band, fill=white, right of=l5, right=.75cm, minimum height=.35cm, minimum width=2.5cm] (l6) {};
                \node[band, fill=black, right of=l6, right=.25cm, minimum height=.35cm, minimum width=1cm] (l7) {};
                \node[band, fill=white, right of=l7, right=-.5cm, minimum height=.35cm, minimum width=4cm] (l8) {};
                
                \node[voidNode, above of=l3, above=-.6cm] (v1) {};
                \draw ([xshift=.25cm,yshift=.1cm]v1.north east) node {\large$v_1$};
                \node[voidNode, above of=l5, above=-.6cm] (v2) {};
                \draw ([xshift=.25cm,yshift=.1cm]v2.north east) node {\large$v_2$};
                \node[band, fill=white, above of=l6, above=-.4cm, minimum width=1cm] (a) {};
                \draw ([xshift=-.25cm,yshift=.1cm]a.south west) node {\large$a$};
                
                \node[left of=l3, right=.72cm, above=-.11cm] (i1) {};
                \node[right of=l3, right=-.72cm, above=-.11cm] (i2) {};
                \node[left of=l5, right=-.88cm, above=-.11cm] (i3) {};
                \node[right of=l5, right=.88cm, above=-.11cm] (i4) {};
                \node[left of=l7, right=.45cm, above=-.11cm] (i5) {};
                \node[left of=l7, right=1.15cm, above=.45cm] (i6) {};
                
                \draw[-] (l3.north west) to[out=90, in=-140] (v1);
                \draw[-] (v1) to[out=-40, in=90] (l3.north east);
                \draw[-] (l5.north west) to[out=60, in=180] (v2);
                \draw[-] (v2) to[out=0, in=120] (l5.north east);
                \draw[-] (l7.north west) to[out=90, in=0] (a.south east);
            \end{tikzpicture}
        \end{adjustbox}
        \caption{Example of graph returned if \begin{math}N_2>f(\lambda)\end{math} and $|\lambda|>1$.}
        \label{fig:alg-chainA2}
    \end{subfigure}
    \begin{subfigure}{0.9\textwidth}
        \centering
        \begin{adjustbox}{max width=\linewidth}
            \begin{tikzpicture}[node distance=1cm]
                \node[] (x) {} ;
                \node[band, fill=black, below of=x, minimum height=.35cm, minimum width=1.5cm] (l1) {};
                \node[band, fill=white, right of=l1, right=-.25cm, minimum height=.35cm, minimum width=2cm] (l2) {};
                \node[band, fill=black, right of=l2, right=0cm, minimum height=.35cm, minimum width=0.5cm] (l3) {};
                \node[band, fill=white, right of=l3, right=-.75cm, minimum height=.35cm, minimum width=3cm] (l4) {};
                \node[band, fill=black, right of=l4, right=.5cm, minimum height=.35cm, minimum width=3.5cm] (l5) {};
                \node[band, fill=white, right of=l5, right=.75cm, minimum height=.35cm, minimum width=2.5cm] (l6) {};
                \node[band, fill=black, right of=l6, right=.25cm, minimum height=.35cm, minimum width=1cm] (l7) {};
                \node[band, fill=white, right of=l7, right=-.5cm, minimum height=.35cm, minimum width=4cm] (l8) {};
                \node[rectangle, pattern={Lines[angle=45,yshift=.5pt]}, draw=none, minimum width=.8cm, minimum height=.35cm, left of=l6, right=1.4cm] (ig1) {} ;
                \node[rectangle, pattern={Lines[angle=45,yshift=.5pt]}, draw=none, minimum width=1.6cm, minimum height=.35cm, left of=l8, right=-1.25cm] (ig1) {} ;
                
                \node[voidNode, above of=l7, above=-.6cm] (v1) {};
                \draw ([xshift=.25cm,yshift=.1cm]v1.north east) node {\large$v_1$};
                
                \node[voidNode, left of=l7, right=-.7cm] (i1) {};
                \node[voidNode, right of=l7, right=.75cm] (i2) {};
                
                \draw[-] (i1) to[out=70, in=180] (v1);
                \draw[-] (v1) to[out=0, in=110] (i2);
            \end{tikzpicture}
        \end{adjustbox}
        \caption{Example of graph returned if \begin{math}N_2<f(\lambda)\end{math} and $|\lambda|=2$.}
        \label{fig:alg-ign1}
    \end{subfigure}
    \begin{subfigure}{0.9\textwidth}
        \centering
        \begin{adjustbox}{max width=\linewidth}
            \begin{tikzpicture}[node distance=1cm]
                \node[] (x) {} ;
                \node[band, fill=black, below of=x, minimum height=.35cm, minimum width=1.5cm] (l1) {};
                \node[band, fill=white, right of=l1, right=-.25cm, minimum height=.35cm, minimum width=2cm] (l2) {};
                \node[band, fill=black, right of=l2, right=0cm, minimum height=.35cm, minimum width=0.5cm] (l3) {};
                \node[band, fill=white, right of=l3, right=-.75cm, minimum height=.35cm, minimum width=3cm] (l4) {};
                \node[band, fill=black, right of=l4, right=.5cm, minimum height=.35cm, minimum width=3.5cm] (l5) {};
                \node[band, fill=white, right of=l5, right=.75cm, minimum height=.35cm, minimum width=2.5cm] (l6) {};
                \node[band, fill=black, right of=l6, right=.25cm, minimum height=.35cm, minimum width=1cm] (l7) {};
                \node[band, fill=white, right of=l7, right=-.5cm, minimum height=.35cm, minimum width=4cm] (l8) {};
                \node[rectangle, pattern={Lines[angle=45,yshift=.5pt]}, draw=none, minimum width=1.6cm, minimum height=.35cm, left of=l4, right=-.5cm] (ig1) {} ;
                
                \node[voidNode, above of=l3, above=-.6cm] (v1) {};
                \draw ([xshift=.25cm,yshift=.1cm]v1.north east) node {\large$v_1$};
                \node[voidNode, above of=l5, above=-.6cm] (v2) {};
                \draw ([xshift=.25cm,yshift=.1cm]v2.north east) node {\large$v_2$};
                
                \node[left of=l3, right=.68cm, above=-.11cm] (i1) {};
                \node[voidNode, right of=l3, right=.7cm] (i2) {};
                \node[left of=l5, right=-.88cm, above=-.11cm] (i3) {};
                \node[right of=l5, right=.88cm, above=-.11cm] (i4) {};
                
                \draw[-] (l3.north west) to[out=70, in=-130] (v1);
                \draw[-] (v1) to[out=0, in=110] (i2);
                \draw[-] (l5.north west) to[out=60, in=180] (v2);
                \draw[-] (v2) to[out=0, in=120] (l5.north east);
            \end{tikzpicture}
        \end{adjustbox}
        \caption{Example of graph returned if \begin{math}N_2<f(\lambda)\end{math} and $|\lambda|>2$.}
        \label{fig:alg-ign2}
    \end{subfigure}
    \caption{Illustration of different graphs returned by Algorithm $\mathcal{A}_1$, given the same path $L_1$ and varying $N_2$ and $\lambda$. White paths are of colour $c$, and their hatched parts belong to the set $V_{1-2}$ of ignored vertices. Figures \ref{fig:alg-chainA1} and \ref{fig:alg-chainA2} represent the execution of Lines \ref{alg:line:a} to \ref{alg:line:alink3}. Figures \ref{fig:alg-ign1} and \ref{fig:alg-ign2} represent the execution of Lines \ref{alg:line:y2} to \ref{alg:line:y3link}. Figures \ref{fig:alg-chainA2} and \ref{fig:alg-ign2} represent the execution of Lines \ref{alg:line:chained}.}
    \label{fig:alg-ex}
\end{figure}

The proof is organised in three stages. 
We first establish the direct properties of the construction: connectivity, order, and the structure of its monochromatic components. 
We then characterise the only ways in which the construction can create an additional occurrence of $L_1$ and show how every such obstruction can be repaired. 
Finally, we prove that at least one optimal solution has the canonical form generated by $\mathcal{A}_1$.

Hereafter, for any graph $H$ returned by Algorithm~$\mathcal{A}_1$, we denote by $H_1$ the subpath constructed at Line~\ref{alg:line:l1}, and by $a$ the path constructed at Line~\ref{alg:line:a}.
Observe that, unlike the notation used for solution graphs, $H_1$ here refers to the canonical copy of $L_1$ created by Algorithm~$\mathcal{A}_1$, which may not be unique in $H$.

\subsubsection{Direct Properties of the Construction}

\begin{lemma}\label{lem:algo1}
    Given an instance in $\mathcal{J}$ with set $\Lambda$ and a non-empty subset $\lambda\subseteq\Lambda$, Algorithm~$\mathcal{A}_1$ returns a connected graph. 
\end{lemma}
\begin{proof}
    Let $H$ be the graph returned by Algorithm~$\mathcal{A}_1$. 
    By Line~\ref{alg:line:l1}, graph $H$ includes a subpath $H_1\approx L_1$. 
    Let $x=f(\lambda)$, computed at Line~\ref{alg:line:x}.
    
    If $|\lambda|=1$, then there is no vertex created in Line~\ref{alg:line:v}.
    Otherwise, since $|\lambda|>0$, graph $H$ includes $|\lambda|-1$ vertices $v_1, \ldots, v_{|\lambda|-1}$ created in Line~\ref{alg:line:v}. 
    We will show that the vertices created in Line~\ref{alg:line:v} are connected to the subpath $H_1$.
    In every case comparing $x$ to $N_2$, vertex $v_1$ is adjacent to a vertex in $H_1$ either by Lines~\ref{alg:line:y1link},~\ref{alg:line:y2link} or~\ref{alg:line:y3link}.
    Moreover, Line~\ref{alg:line:chained} ensures that every vertex $v_i$, $2\leq i<|\lambda|$, is adjacent to two vertices in $H_1$.
    Hence, all vertices created in Line~\ref{alg:line:v} belong to the same connected component as $H_1$.
    
    If $N_2>x$, Algorithm~$\mathcal{A}_1$ creates a path $a$ of $N_2-x$ vertices at Line~\ref{alg:line:a}. 
    If $|\lambda|=1$, Line~\ref{alg:line:alink1} connects one extremity of $a$ to a vertex in $H_1$; otherwise, this is done at Line~\ref{alg:line:alink2} or~\ref{alg:line:alink3}.
    Hence, every vertex of $a$ belongs to the same connected component as $H_1$.
    Therefore, graph $H$ is connected.
\end{proof}

Having established the connectivity of $H$, we now study the size $|V(H)|$ for the graph $H$ returned by Algorithm~$\mathcal{A}_1$ running with a non-empty subset $\lambda$.

\begin{lemma}\label{lem:Hsize}
    Given an instance in $\mathcal{J}$ with its set $\Lambda$ and a non-empty subset $\lambda$ of $\Lambda$, Algorithm~$\mathcal{A}_1$ returns a graph with \begin{math}N_1+|\lambda|-1+\max\{0,N_2-f(\lambda)\}\end{math} vertices. 
\end{lemma}
\begin{proof}
    Let $H$ be the graph returned by Algorithm~$\mathcal{A}_1$. 
    According to Line~\ref{alg:line:l1}, graph $H$ includes one subpath $H_1\approx L_1$ with $N_1$ vertices. 

    Since set $\lambda$ is non-empty, Line~\ref{alg:line:v} adds $|\lambda|-1\geq0$ vertices to $H$.
    If \begin{math}N_2>f(\lambda)\end{math}, Line~\ref{alg:line:a} adds $N_2-f(\lambda)$ vertices to $H$.
    If \begin{math}N_2\leq f(\lambda)\end{math}, no other vertex is added to $H$.
    Hence, $H$ contains \begin{math}N_1+|\lambda|-1+\max\{0,N_2-f(\lambda)\}\end{math} vertices.
\end{proof}

Since every perfect subset satisfies Condition~\ref{cond:notempty} of Definition~\ref{def:perfect}, Algorithm~$\mathcal{A}_1$ returns, for every perfect subset $\lambda$, a connected graph with \begin{math}N_1+|\lambda|-1+\max\{0,N_2-f(\lambda)\}\end{math} vertices.

We now show that the vertices of colour $c$ in $H$ induce a forest containing a distinguished tree associated with $\lambda$.

\begin{lemma}\label{lem:tree}
    Given an instance in $\mathcal{J}$ with its set $\Lambda$ and a perfect subset $\lambda\subseteq\Lambda$, let $H$ be the graph returned by Algorithm~$\mathcal{A}_1$. 
    The subgraph induced by the vertices of colour $c$ in $H$ consists of : 
    \begin{enumerate}[(i)]
        \item one connected component $T_\lambda$, which is a tree containing the subpaths of $\lambda$, the vertices $v_1,\ldots,v_{|\lambda|-1}$, and the path $a$;
        \item for each subpath $l\in\Lambda\setminus\lambda$, one connected component consisting only of $l$.
    \end{enumerate}
\end{lemma}
\begin{proof}
    Let $H_c$ be the subgraph of $H$ induced by the vertices in $H$ coloured in $c$.
    Denote by $n$, $m$ and $p$ respectively the number of vertices, edges and connected components in $H_c$.
    Its cyclomatic number $\mu$ is $\mu=m-n+p$.
    Since $\mu=0$ if and only if every connected component is acyclic, proving $\mu=0$ shows that each component of $H_c$ is a tree or a path.
    
    According to Line~\ref{alg:line:l1}, subgraph $H_c$ initially consists of the subpaths of $\Lambda$ contained in $H_1$.
    If $|\lambda|>1$, Line~\ref{alg:line:v} additionally creates the vertices $v_1,\ldots,v_{|\lambda|-1}$ of colour $c$.
    Hence, at this stage, 
    \begin{math}n=\sum_{l_i\in\Lambda}|l_i| +|\lambda|-1,\end{math}
    \begin{math}m=\sum_{l_i\in\Lambda}(|l_i|-1)\end{math} and 
    \begin{math}p=|\Lambda|+|\lambda|-1.\end{math}
    Therefore, 
    \begin{displaymath}
        \mu=\sum_{l_i\in\Lambda}(|l_i|-1)-\sum_{l_i\in\Lambda}|l_i| -(|\lambda|-1)+|\Lambda|+|\lambda|-1=0.
    \end{displaymath}
    
    If \begin{math}N_2>f(\lambda)\end{math}, Line~\ref{alg:line:a} creates a path $a$ with $|a|=N_2-f(\lambda)$ vertices of colour $c$. 
    By Lines~\ref{alg:line:alink1},~\ref{alg:line:alink2} and~\ref{alg:line:alink3}, one extremity of $a$ is adjacent to a vertex of a subpath in $\lambda$.
    Thus, adding $a$ increases both $n$ and $m$ by $|a|$, while leaving $p$ unchanged, so $\mu$ remains equal to $0$.
    
    We now consider the vertices $v_1,\ldots,v_{|\lambda|-1}$.
    Since $\lambda$ satisfies Condition~\ref{cond:ign}, the vertices of $l_1$ and $l_2$ selected in Lines~\ref{alg:line:y1link},~\ref{alg:line:y2link} and~\ref{alg:line:y3link} exist in $l_1$ and $l_2$.
    If $|\lambda|>1$, one of these lines connects $v_1$ to one vertex in $l_1$ and one vertex in $l_2$.
    Then, by Line~\ref{alg:line:chained}, for each $1< i\leq|\lambda|-1$, vertex $v_i$ is adjacent to the last vertex of $l_i$ and the first vertex of $l_{i+1}$. 
    
    By induction on $i$, after processing vertices $v_1,\ldots,v_i$, the subpaths $l_1,\ldots,l_{i+1}$ together with $v_1,\ldots,v_i$ belong to the same connected component $H_c$.
    Moreover, each added edge joins two distinct connected components, decreasing $p$ by one while increasing $m$ by one. Hence, $\mu$ remains equal to $0$ throughout the construction.

    Consequently, every connected component of $H_c$ is acyclic.
    Furthermore, all subpaths of $\lambda$, the vertices $v_1,\ldots,v_{|\lambda|-1}$ and, when non-empty, the path $a$, belong to the same connected component in $H_c$. This component is therefore a tree, denoted $T_\lambda$.
    No vertex is connected to a subpath of $\Lambda\setminus\lambda$, so each subpath in $\Lambda\setminus\lambda$ forms a distinct connected component of $H_c$.
\end{proof}

By Lemma~\ref{lem:tree}, the vertices of colour $c$ in a graph returned by Algorithm~$\mathcal{A}_1$ induce a unique tree, denoted $T_\lambda$, containing the subpaths of $\lambda$, the vertices $v_1,\ldots,v_{|\lambda|-1}$ and the path $a$.
Vertices of degree greater than two in $T_\lambda$ are called \emph{junction} vertices.
Depending on the values $N_2$, $|\lambda|$ and $f(\lambda)$, the tree $T_\lambda$ contains at most two junction vertices, as illustrated in Figure~\ref{fig:alg-ex}. A junction vertex has degree at most $4$ when it is the only \emph{used} vertex in its subpath.
Given a junction vertex $j$, each connected component of $T_\lambda \setminus \{j\}$ defines a \emph{branch} of $j$, namely each longest subpath of colour $c$ in this component with one extremity adjacent to $j$. Such the longest subpath may not be unique if the component contains another junction vertex.

Figure~\ref{fig:pfork} illustrates a tree $T$ with a unique junction vertex $u$ of degree $3$. Its three branches are denoted by $p$, $q$ and $r$. 

\begin{figure}[ht]
    \centering
    \begin{adjustbox}{max width=0.55\textwidth}
        \begin{tikzpicture}[node distance=1cm]
            \node[band, fill=white, minimum width=5cm] (l1) {};
            \node[voidNode, right of=l1, right=2cm] (x1) {};
            \draw ([xshift=.1cm,yshift=.3cm]x1.north west) node {\Large$u$};
            \node[band, above of=x1, right=1cm, fill=white, minimum width=4cm] (l2) {};
            \node[band, below of=x1, right=1cm, fill=white, minimum width=3cm] (l3) {};

            \draw[-] (x1) -- (l1);
            \draw[-] (x1) to[out=30,in=180] (l2);
            \draw[-] (x1) to[out=-30,in=180] (l3);
            
            \draw [decorate,decoration={brace,amplitude=5pt,raise=0pt}]
            ([yshift=5pt]l1.north west) -- ([yshift=5pt]l1.north east) 
            node [black, midway, yshift=12pt] {\Large$p$};
            \draw [decorate,decoration={brace,amplitude=5pt,raise=0pt}]
            ([yshift=5pt]l2.north west) -- ([yshift=5pt]l2.north east) 
            node [black, midway, yshift=12pt] {\Large$q$};
            \draw [decorate,decoration={brace,mirror,amplitude=5pt,raise=0pt}]
            ([yshift=-5pt]l3.south west) -- ([yshift=-5pt]l3.south east) 
            node [black, midway, yshift=-12pt] {\Large$r$};
        \end{tikzpicture}
    \end{adjustbox}
    \caption{Illustration of a tree of vertices of colour $c$ with three branches of $p$, $q$ and $r$ vertices all having an extremity linked to the vertex $u$.}
    \label{fig:pfork}
\end{figure}
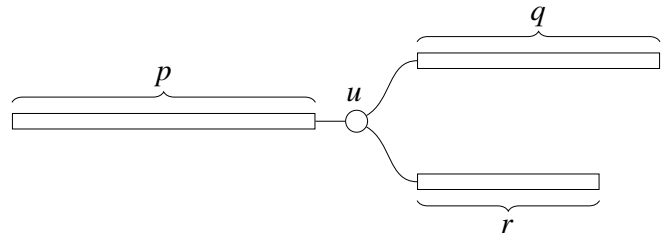

Based on the previous result, we elaborate on the structure of the subpaths of colour $c$ within graph $H$ returned by Algorithm~$\mathcal{A}_1$. 

\begin{lemma}\label{lem:l2-added}
    Given an instance in $\mathcal{J}$ with its set $\Lambda$ and a perfect subset $\lambda\subseteq\Lambda$ such that \begin{math}N_2\geq f(\lambda)\end{math}, let $H$ be the graph returned by Algorithm~$\mathcal{A}_1$. 
    Then, $H$ contains a unique maximum subpath of colour $c$, and this subpath contains exactly $N_2$ vertices.
\end{lemma}
\begin{proof}
    By Lemma~\ref{lem:tree}, the vertices of colour $c$ in $H$ induce a collection of paths together with a tree $T_\lambda$ containing all the subpaths of $\lambda$ as well as all the vertices $v_i$ for $1\leq i\leq |\lambda|-1$. 
    Moreover, since \begin{math}N_2\geq f(\lambda)\end{math}, $T_\lambda$ contains a path $a$, created at Line~\ref{alg:line:a}, with $N_2-f(\lambda)$ vertices.
    Hence, $|V(T_\lambda)| = N_2$.

    In addition, each vertex $v_i$ is adjacent to the last vertex of $l_i\in\lambda$ and the first vertex of $l_{i+1}\in\lambda$. Similarly, one extremity of the path $a$ is adjacent to exactly one extremity of a subpath in $\lambda$. 
    Moreover, each extremity of a subpath in $\lambda$ is used by at most one such adjacency. Hence, no vertex of $T_\lambda$ has degree greater than two. Since $T_\lambda$ is connected, it follows that $T_\lambda$ is itself a path.

    Finally, because the instance belongs to $\mathcal{J}$, the sample $\{L_1,L_2\}$ is \emph{without inclusion}. Hence, every subpath $l\in\Lambda$ satisfies $|l|<N_2$.
    Therefore, every monochromatic subpath of colour $c$ in $H$ distinct from $T_\lambda$ contains strictly fewer than $N_2$ vertices.

    Consequently, $T_\lambda$ is the unique maximum subpath of colour $c$ in $H$.
\end{proof}

\begin{lemma}\label{lem:l2-ign}
    Given an instance in $\mathcal{J}$ with its set $\Lambda$ and a perfect subset $\lambda$ of $\Lambda$ with $N_2< f(\lambda)$, let $H$ be the graph returned by $\mathcal{A}_1$. 
    The maximum subpath of colour $c$ in $H$ is unique and contains $N_2$ vertices. 
\end{lemma}
\begin{proof}
    By Lemma~\ref{lem:tree}, the subgraph of $H$ induced by the vertices of colour $c$ consists of a collection of paths together with a tree $T_\lambda$ containing the subpaths of $\lambda$ and the vertices created by $\mathcal{A}_1$.

    Since the sample $\{L_1,L_2\}$ is \emph{without inclusion}, every subpath $l\in\Lambda$ satisfies $|l|<N_2$. Moreover, since $\lambda$ is perfect and \begin{math}N_2<f(\lambda)\end{math}, no path $a$ is added to $H$ and $|\lambda|>1$.
    
    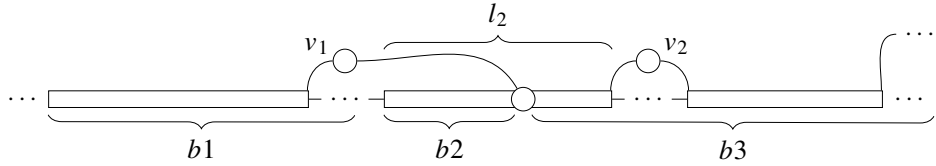
\begin{figure}[ht]
    \centering
    \begin{adjustbox}{max width=.8\textwidth}
        \begin{tikzpicture}[node distance=1cm]
            \node (x0) {$\ldots$};
            \node[band, right of=x0, right=-.6cm, fill=white, minimum width=4cm, label=above:\large$l_1$] (l1) {};
            \node[right of=l1, right=1.2cm] (x1) {$\ldots$};
            \node[band, right of=x1, right=-.4cm, fill=white, minimum width=3.5cm] (l2) {};
            \node[right of=l2, right=.95cm] (x2) {$\ldots$};
            \node[band, right of=x2, right=-.4cm, fill=white, minimum width=3cm, label=above:\large$l_3$] (l3) {};
            \node[right of=l3, right=.55cm] (x3) {$\ldots$};
            
            \node[voidNode, above of=x1, above=-.6cm] (v1) {};
            \draw ([xshift=-.3cm,yshift=.1cm]v1.north west) node {\large$v_1$};
            \node[voidNode, above of=x2, above=-.6cm] (v2) {};
            \draw ([xshift=.3cm,yshift=.1cm]v2.north east) node {\large$v_2$};
            \node[above of=x3, above=-.3cm, right=-.2cm] (v3) {$\ldots$};
            
            \node[right of=l1, right=.85cm] (i1) {};
            \node[voidNode, right of=l2, right=-.8cm] (i2) {};
            \node[right of=l2, right=.6cm] (i3) {};
            \node[left of=l3, left=.35cm] (i4) {};
            \node[right of=l3, right=.35cm] (i5) {};
            
            \draw[-] (l1) -- (x1);
            \draw[-] (x1) -- (l2);
            \draw[-] (l2) -- (x2);
            \draw[-] (x2) -- (l3);
            
            \draw[-] (i1) to[out=90, in=180] (v1);
            \draw[-] (v1) to[out=0, in=120] (i2);
            \draw[-] (i3) to[out=90, in=180] (v2);
            \draw[-] (v2) to[out=0, in=90] (i4);
            \draw[-] (i5) to[out=90, in=180] (v3);
            
            \draw [decorate,decoration={brace,amplitude=6pt,raise=5pt}] 
            ([yshift=13pt]l2.north west) -- ([yshift=13pt]l2.north east) 
            node [black,midway,yshift=20pt] {\large$l_2$};            
            \draw [decorate,decoration={brace,mirror,amplitude=6pt,raise=0pt}] 
            ([yshift=-3pt]l1.south west) -- ([yshift=-20pt, xshift=.5pt]v1.south east) 
            node [black,midway,yshift=-15pt] {\large$b1$};
            \draw [decorate,decoration={brace,mirror,amplitude=6pt,raise=0pt}] 
            ([yshift=-3pt]l2.south west) -- ([yshift=-3pt]i2.south west) 
            node [black,midway,yshift=-15pt] {\large$b2$};
            \draw [decorate,decoration={brace,mirror,amplitude=6pt,raise=0pt}] 
            ([yshift=-3pt]i2.south east) -- ([yshift=-3pt, xshift=.5pt]x3.south east) 
            node [black,midway,yshift=-15pt] {\large$b3$};
        \end{tikzpicture}
    \end{adjustbox}
    \caption{Illustration of a solution graph $H$ for one instance of $\mathcal{J}$. The set $\lambda_H$ contains more than two subpaths with $f(\lambda)-N_2$ ignored vertices. There are at least two additional vertices $v_1$ and $v_2$.}
    \label{fig:allInOne}
\end{figure}

    We first assume that $|\lambda|>2$. 
    By Lines~\ref{alg:line:chained} and~\ref{alg:line:y3link}, the tree $T_\lambda$ contains a unique junction vertex, namely the $(f(\lambda)-N_2+1)$-th vertex of $l_2$ (see Figure~\ref{fig:allInOne}). 
    Condition~\ref{cond:ign} ensures that this vertex exists.

    This junction vertex induces three branches, denoted $b_1$, $b_2$ and $b_3$ in Figure~\ref{fig:allInOne}. 
    Branch $b_2$ contains exactly $f(\lambda)-N_2$ vertices, while branches $b_1$ and $b_3$ contain respectively $|l_1|+1$ and at least $|l_3|+1$ vertices. 
    Since $\lambda$ satisfies Condition~\ref{cond:ign}, we have
    \begin{displaymath}
        |l_1|+1>f(\lambda)-N_2\qquad \text{and}\qquad |l_3|+1>f(\lambda)-N_2.
    \end{displaymath}
    Hence, $b_2$ is strictly the unique shortest branch.

    Any maximum subpath of colour $c$ in $H$ is obtained by taking the two longest branches incident to the junction vertex. Therefore, the unique maximum subpath of colour $c$ is obtained by excluding $b_2$. Its number of vertices is 
    \begin{displaymath}
        \left(\sum_{l_i\in\lambda}|l_i|\right)+|\lambda|-1-(f(\lambda)-N_2)=N_2.
    \end{displaymath}
    Consequently, the maximum subpath of colour $c$ in $H$ is unique and contains $N_2$ vertices.
    
    Secondly, assume that $|\lambda|=2$. 
    By Line~\ref{alg:line:y2link}, vertex $v_1$ is adjacent to the $(|l_1|-y)$-th vertex of $l_1$ and the $(y'+1)$-th vertex of $l_2$, where 
    \begin{displaymath}
        y=\min\left\{f(\lambda)-N_2,\left\lfloor\frac{|l_1|}{2}\right\rfloor-1\right\}\qquad \text{and}\qquad y'=f(\lambda)-N_2-y.
    \end{displaymath}
    Hence, $y+y'=f(\lambda)-N_2$.

    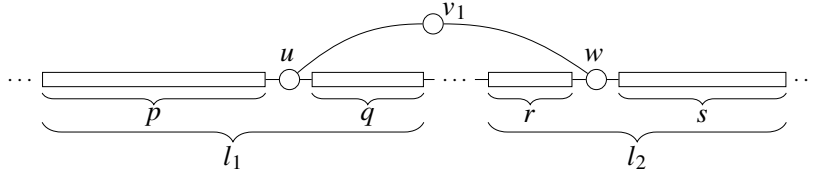
\begin{figure}[ht]
    \centering
    \begin{adjustbox}{max width=.7\textwidth}
        \begin{tikzpicture}[node distance=1cm]
            \node (x0) {$\ldots$};
            \node[band, right of=x0, right=-.6cm, fill=white, minimum width=4cm] (l1) {};
            \node[voidNode, right of=l1, right=1.25cm] (x1) {};
            \draw ([xshift=.1cm,yshift=.3cm]x1.north west) node {\Large$u$};
            \node[band, right of=x1, right=-.6cm, fill=white, minimum width=2cm] (l2) {};
            \node[right of=l2, right=.2cm] (x2) {$\ldots$};
            \node[band, right of=x2, right=-.4cm, fill=white, minimum width=1.5cm] (l3) {};
            \node[voidNode, right of=l3, right=0cm] (x3) {};
            \draw ([xshift=.1cm,yshift=.3cm]x3.north west) node {\Large$w$};
            \node[band, right of=x3, right=-.6cm, fill=white, minimum width=3cm] (l4) {};
            \node[right of=l4, right=.5cm] (x4) {$\ldots$};
            \node[voidNode, above of=x2, above=-.35cm, left=.2cm] (v1) {};
            \draw ([xshift=.25cm,yshift=.1cm]v1.north east) node {\Large$v_1$};
    
            \draw[-] (l1) -- (x1);
            \draw[-] (x1) -- (l2);
            \draw[-] (l2) -- (x2);
            \draw[-] (x2) -- (l3);
            \draw[-] (l3) -- (x3);
            \draw[-] (x3) -- (l4);
            
            \draw[-] (x1) to[out=40, in=180] (v1);
            \draw[-] (v1) to[out=0, in=145] (x3);
            
            \draw [decorate,decoration={brace,mirror,amplitude=6pt,raise=0pt}]
            ([yshift=-3pt]l1.south west) -- ([yshift=-3pt]l1.south east) 
            node [black,midway,yshift=-12pt] {\Large$p$};
            \draw [decorate,decoration={brace,mirror,amplitude=6pt,raise=0pt}]
            ([yshift=-3pt]l2.south west) -- ([yshift=-3pt]l2.south east) 
            node [black,midway,yshift=-12pt] {\Large$q$};
            \draw [decorate,decoration={brace,mirror,amplitude=10pt,raise=5pt}]
            ([yshift=-13pt]l1.south west) -- ([yshift=-13pt]l2.south east) 
            node [black,midway,yshift=-23pt] {\Large$l_1$};
            \draw [decorate,decoration={brace,mirror,amplitude=6pt,raise=0pt}]
            ([yshift=-3pt]l3.south west) -- ([yshift=-3pt]l3.south east) 
            node [black,midway,yshift=-12pt] {\Large$r$};
            \draw [decorate,decoration={brace,mirror,amplitude=6pt,raise=0pt}]
            ([yshift=-3pt]l4.south west) -- ([yshift=-3pt]l4.south east) 
            node [black,midway,yshift=-12pt] {\Large$s$};
            \draw [decorate,decoration={brace,mirror,amplitude=10pt,raise=5pt}]
            ([yshift=-13pt]l3.south west) -- ([yshift=-13pt]l4.south east) 
            node [black,midway,yshift=-23pt] {\Large$l_2$};
    
        \end{tikzpicture}
    \end{adjustbox}
    \caption{Illustration of a solution graph $H$ for one instance of $\mathcal{J}$. The set $\lambda_H$ contains two subpaths $l_1$ and $l_2$ with $p+q+1$ and $r+s+1$ vertices, respectively. There is one additional vertex $v_1$.}
    \label{fig:allInTwo}
\end{figure}

    We distinguish two cases.

    Assume first that \begin{math}f(\lambda)-N_2>\left\lfloor\frac{|l_1|}{2}\right\rfloor-1.\end{math}
    Then \begin{math}y=\left\lfloor\frac{|l_1|}{2}\right\rfloor-1.\end{math}
    Since $\lambda$ satisfies Condition~\ref{cond:twochains}, we obtain 
    \begin{math}
        y'\leq\left\lfloor\frac{|l_2|}{2}\right\rfloor-1.
    \end{math}
    Therefore, both neighbours of $v_1$, denoted $u$ and $w$ in Figure~\ref{fig:allInTwo}, are junction vertices of $T_\lambda$.

    Let $p,q$ be the branches incident to $u$ in $l_1$, and $r,s$ those incident to $w$ in $l_2$. 
    By construction, $|q|=y$ and $|r|=y'$, while 
    \begin{math}|p|=|l_1|-y-1>|q|\end{math} and \begin{math}|s|=|l_2|-y'-1>|r|.\end{math}
    Hence, $q$ and $r$ are the shortest branches incident to the junction vertices. 
    If we excluded both $q$ and $r$, the number of vertices in the subpath of colour $c$ obtained from $p$, $v_1$ and $s$, is 
    \begin{displaymath}
        |l_1|+|l_2|+1-y-y'=f(\lambda)-(f(\lambda)-N_2)=N_2. 
    \end{displaymath}
    Since every subpath contained entirely in $l_1$ or $l_2$ has fewer than $N_2$ vertices, this subpath is the unique maximum subpath of colour $c$ in $H$.

    Assume now that \begin{math}f(\lambda)-N_2\leq\left\lfloor\frac{|l_1|}{2}\right\rfloor-1.\end{math}
    Then \begin{math}y=f(\lambda)-N_2\end{math} and $y'=0$, so the neighbour of $v_1$ in $l_1$ is the unique junction vertex of $T_\lambda$.
    Its three branches contain respectively $|l_1|-y-1$, $y$ and $|l_2|+1$ vertices. 
    Since $\lambda$ satisfies Condition~\ref{cond:ign}, \begin{math}|l_2|+1>f(\lambda)-N_2=y\end{math}, and since \begin{math}y\leq\left\lfloor\frac{|l_1|}{2}\right\rfloor-1\end{math}, we also obtain 
    \begin{displaymath}
        |l_1|-y-1\geq\left\lceil\frac{|l_1|}{2}\right\rceil>y.
    \end{displaymath}
    Thus, the branch of size $y$ is the unique shortest branch. Therefore, the unique maximum subpath of colour $c$ is obtained from the two other branches and its order is 
    \begin{displaymath}
        |l_1|+|l_2|+1-y=f(\lambda)-(f(\lambda)-N_2)=N_2.
    \end{displaymath}

    Therefore, $H$ contains a unique maximum subpath of colour $c$, and this subpath contains exactly $N_2$ vertices.
\end{proof}

We now consider the subpath isomorphic to $L_1$ and its conditions of existence in a graph returned by Algorithm~$\mathcal{A}_1$ defined below.

\subsubsection{Additional Occurrences and Their Repair}

The preceding lemmas guarantee the desired copy of $L_2$, but the copy of $L_1$ inserted at the first line of $\mathcal{A}_1$ need not yet be unique. 
The following lemmas characterise the possible copies of $L_1$ in the graph, starting with the necessary conditions for a set of vertices to induce a subgraph isomorphic to $L_1$.

\begin{lemma}\label{lem:l1inH}
    Given a graph $H$ constructed by Algorithm~$\mathcal{A}_1$ and a subset $W$ of vertices in $H$, that induces a subpath isomorphic to $L_1$, then it satisfies Conditions~\ref{wN1},~\ref{wch} and~\ref{wnonc} :
    \begin{enumerate}[(a)]
        \setcounter{enumi}{3}
        \item subset $W$ contains exactly $N_1$ vertices ; \label{wN1}
        \item vertices in $W$ induce a subpath in $H$ ; \label{wch}
        \item subset $W$ contains all the vertices not coloured in $c$ in $H$. \label{wnonc}
    \end{enumerate}
\end{lemma}
\begin{proof}
    Because $H_W\approx L_1$, the colour-preserving isomorphism gives \begin{math}|W|=|V(L_1)|=N_1\end{math}, and $H[W]$ is an induced path.
    Algorithm~$\mathcal{A}_1$ adds only vertices of colour $c$ to the canonical copy $H_1$.
    Hence, $H$ and $L_1$ contain the same number of vertices whose colour differs from $c$.
    Every such vertex must occur in any coloured copy of $L_1$, so all vertices of $H$ not coloured in $c$ belong to $W$.
    Thus, Conditions~\ref{wN1},~\ref{wch}, and~\ref{wnonc} hold.
\end{proof}

We now consider the cases where a set of vertices inducing a subgraph isomorphic to $L_1$ contains vertices outside $V(H_1)$.

\begin{lemma}\label{lem:l1wVi}
    Given an instance in $\mathcal{J}$ with its set $\Lambda$ and a perfect subset $\lambda$ of $\Lambda$ with $|\lambda|>1$, let $H$ be the graph returned by $\mathcal{A}_1$.
    We consider a subset of vertices $W$ in graph $H$ that contains at least one vertex $v_i$ with 
    \begin{math}1\leq i\leq |\lambda|-1\end{math} created at Line~\ref{alg:line:v}.
    Subset $W$ inducing a subpath $H_W\approx L_1$ implies that $\lambda$ satisfies the following Conditions~\ref{w2ext},~\ref{wf3} and~\ref{wl3} :
    \begin{enumerate}[(a)]
        \setcounter{enumi}{6}
        \item $|\lambda|=2$ with $l_1$ and $l_2$ containing the extremities of path $H_1$ ; \label{w2ext}
        \item \begin{math}N_2>f(\lambda)\end{math} and $f(\lambda)=3$ ; \label{wf3}
        \item there exists a subpath \begin{math}l\in\Lambda\setminus\lambda\end{math} with $|l|=3$ ; \label{wl3}
    \end{enumerate}
\end{lemma}
\begin{proof}
    \begin{figure}[ht]
    \centering
    \begin{subfigure}{0.5\textwidth}
        \centering
        \begin{adjustbox}{max width=\linewidth}
            \begin{tikzpicture}[node distance=1 cm]
                \node[] (x1) {$\ldots$} ;
                \draw ([xshift=-.25cm,yshift=-.2cm]x1.south east) node {$p_u$};
                \node[voidNode, right of=x1] (x2) {} ;
                \draw ([xshift=-.15cm,yshift=-.2cm]x2.south east) node {$u$};
                \node[band, right of=x2, right=-.5cm, minimum width=2cm] (p) {} ;
                \draw ([xshift=1cm,yshift=-.2cm]p.south west) node {$p$};
                \node[band, right of=x2, fill=black] (p_notc) {} ;
                \node[voidNode, above of=p, below=0.2cm] (x3) {} ;
                \draw ([xshift=.25cm,yshift=.1cm]x3.north east) node {$v_i$};
                \node[voidNode, right of=p, right=.3cm] (x4) {} ;
                \draw ([xshift=-.15cm,yshift=-.2cm]x4.south east) node {$w$};
                \node[right of=x4] (x5) {$\ldots$} ;
                \draw ([xshift=-.25cm,yshift=-.2cm]x5.south east) node {$p_w$};
                
                \draw[-] (x1) -- (x2) ;
                \draw[-] (x2) to[out=70, in=-170] (x3) ;
                \draw[-] (x3) to[out=-10, in=110] (x4) ;
                \draw[-] (x4) -- (x5) ;
                \draw[-] (x2) -- (p) ;
                \draw[-] (x4) -- (p) ;
            \end{tikzpicture}
        \end{adjustbox}
        \caption{Elementary cycle that contains no extremity of $H_1$.}
        \label{fig:cycle-middle}
    \end{subfigure}\hfill
    \begin{subfigure}{0.45\textwidth}
        \centering
        \begin{adjustbox}{max width=\linewidth}
            \begin{tikzpicture}[node distance=0.6 cm]
                \node[voidNode] (x1) {} ;
                \draw ([xshift=-.15cm,yshift=-.2cm]x1.south east) node {$u$};
                \node[pointNode, right of=x1] (x2) {} ;
                \node[right of=x2, right=-.2cm] (x3) {$\ldots$} ;
                \draw ([xshift=-.25cm,yshift=-.2cm]x3.south east) node {$p$};
                \node[pointNode, right of=x3, right=0cm] (x4) {} ;
                \node[voidNode, right of=x4] (x5) {} ;
                \draw ([xshift=-.15cm,yshift=-.2cm]x5.south east) node {$w$};
                \node[right of=x5, right=0cm] (x6) {$\ldots$} ;
                \draw ([xshift=-.25cm,yshift=-.2cm]x6.south east) node {$p_w$};
                \node[voidNode, above of=x3] (vi) {} ;
                \draw ([xshift=.25cm,yshift=.1cm]vi.north east) node {$v_i$};
                
                \draw[-] (x1) -- (x2) ;
                \draw[-] (x2) -- (x3) ;
                \draw[-] (x3) -- (x4) ;
                \draw[-] (x4) -- (x5) ;
                \draw[-] (x5) -- (x6) ;
                \draw[-] (x1) to[out=70, in=-170] (vi) ;
                \draw[-] (vi) to[out=-10, in=110] (x5) ;
            \end{tikzpicture}
        \end{adjustbox}
        \caption{Elementary cycle that contains one extremity of $H_1$.}
        \label{fig:cycle-border}
    \end{subfigure}\hfill
    \caption{Illustration for elementary cycle on graph $H_{\lambda}$ obtain with Algorithm $\mathcal{A}_1$}
    \label{fig:l1-cycle}
\end{figure}
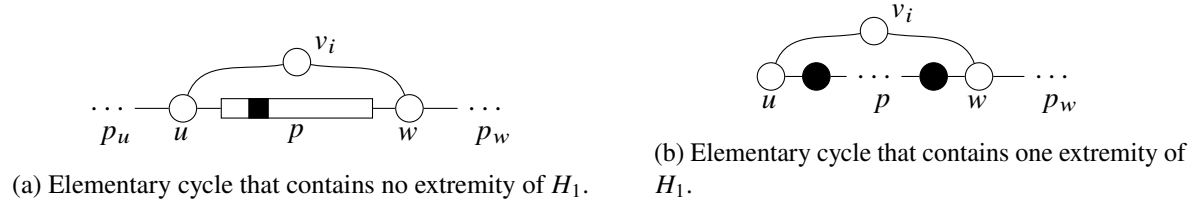
    By Lemma~\ref{lem:l1inH}, $W$ has $N_1$ vertices, induces a path, and contains every vertex of $H_1$ not coloured in $c$.
    Choose a vertex $v_i\in W$, and denote its two neighbours in consecutive members of $\lambda$ by $u$ and $w$.
    Let $p_u$, $p$, and $p_w$ be the portions of $H_1$ lying respectively before $u$, strictly between $u$ and $w$, and after $w$, as in Figure~\ref{fig:cycle-middle}.
    Since $u$ and $w$ belong to distinct subpaths of $\Lambda$, $p$ contains a vertex not coloured in $c$, and every such vertex belongs to $W$.

    Observe that if \begin{math}N_2>f(\lambda)\end{math}, there is a path $a$ and, according to Lines~\ref{alg:line:alink1},~\ref{alg:line:alink2} and \ref{alg:line:alink3}, one of its extremities is adjacent to a vertex of colour $c$ before or after $p$ in $H_1$.
    Observe that whenever $|\lambda|=2$ and \begin{math}N_2<f(\lambda)\end{math}, Lines~\ref{alg:line:y2} and~\ref{alg:line:y2link} apply in the construction of $H$. As $\lambda$ is perfect and any subpath in $\lambda$ has strictly more than \begin{math}f(\lambda)-N_2\geq 1\end{math} vertices, we obtain 
    \begin{displaymath}
        y\leq\left\lfloor{|l_1|}/{2}\right\rfloor-1<|l_1|-1 \qquad \text{and}\qquad y'\leq\left\lfloor{|l_2|}/{2}\right\rfloor-1<|l_2|-1,
    \end{displaymath}
    implying that neither neighbour of $v_1$ are an extremity of $H_1$. 
    
    We first suppose that $W$ exactly one of $u,w$ belongs to $W$, say $u\in W$ and $w\notin W$. 
    Since $H[W]$ is a path which contains $v_i$ and the non-$c$ vertices not coloured in $c$ of $p$, either a path $a$ exists or not, $W$ is confined to $p\cup\{u,v_i\}$.
    Comparing $|W|=|H_1|$ then forces \begin{math}p_u=p_w=\emptyset\end{math} and therefore $u$ and $w$ to be the two extremities of $H_1$.
    It also forces the extremities of $H[W]$ to be $v_i$ and the neighbour of $w$ in $p$.
    By the previous observation, if $u$ and $w$ are extremities of $H_1$ and $|\lambda|=2$, then \begin{math}N_2\geq f(\lambda)\end{math}.  
    Therefore, by Line~\ref{alg:line:y1link}, both extremities of $p$ are not coloured in $c$.
    Hence, one extremity of $H[W]$ is not coloured in $c$, whereas both extremities of $L_1$ are coloured in $c$, a contradiction.
    The case $w\in W$, $u\notin W$ is symmetric.

    Suppose next that $u,w\in W$ and that $W$ contains a vertex of, say, $p_w$.
    To remain an induced path through $v_i$ and the non-$c$ vertices of $p$, the set $W$ contains neither $p_u$ nor the neighbour of $w$ in $p$.
    The equality $|W|=|H_1|$ consequently forces $p_u=\emptyset$ and the latter neighbour to be the only omitted vertex of $H_1$.
    Hence, $u$ is an extremity of $H_1$ and $|l_1|=1$, as in Figure~\ref{fig:cycle-border}.
    By the previous observation, if $u$ is an extremity of $H_1$ and \begin{math}N_2<f(\lambda)\end{math}, then $|\lambda|>2$ and Lines~\ref{alg:line:y3}--\ref{alg:line:y3link} apply in the construction of $H$.
    If \begin{math}N_2<f(\lambda)\end{math}, these Lines and Condition~\ref{cond:ign} give the contradiction 
    \begin{math}1=|l_1|>f(\lambda)-N_2\geq1.\end{math}
    However, if \begin{math}N_2\geq f(\lambda)\end{math}, by Line~\ref{alg:line:y1link}, the omitted neighbour is not coloured in $c$, contradicting Condition~\ref{wnonc}.
    The case where $W$ contains a vertex of $p_u$ instead of $p_w$ is symmetric.

    Hence, $u,w\in W$ and \begin{math}W\cap(p_u\cup p_w)=\emptyset\end{math}.
    If a path $a$ exists, no vertex of the additional path $a$ can belong to $W$: otherwise, by previous observations, path $a$ must be adjacent to either $u$ or $w$, hence the path condition and Line~\ref{alg:line:y1link} force $W$ to omit an adjacent non-$c$ vertex of $p$, contrary to Condition~\ref{wnonc}.
    Thus, \begin{math}W\subseteq p\cup\{u,v_i,w\}\end{math}.
    Comparing again $|W|=|H_1|$ shows that $p_u=p_w=\emptyset$ and that $W$ is obtained from $V(H_1)$ by replacing exactly one vertex of $p$ with $v_i$.
    Therefore, $u$ and $w$ are the extremities of $H_1$, $|\lambda|=2$, and $v_i=v_1$, proving Condition~\ref{w2ext}.
    By the previous observation, if $u$ is an extremity of $H_1$ and $|\lambda|=2$, then \begin{math}N_2\geq f(\lambda)\end{math}.
    Therefore, by Line~\ref{alg:line:y1link}, both extremities of $p$ are not coloured in $c$ which gives $|l_1|=|l_2|=1$, and hence $f(\lambda)=3$.
    The omitted vertex of $p$ must be coloured in $c$; together with the two extremities of $H[W]$ coloured in $c$, it forms a maximal subpath $l\in\Lambda\setminus\lambda$ of length three.
    Since every member of $\Lambda$ has length strictly less than $N_2$, we obtain \begin{math}N_2>3=f(\lambda)\end{math}.
    This proves Conditions~\ref{wf3} and~\ref{wl3}.
\end{proof}

\begin{lemma}\label{lem:l1wA}
    Given an instance in $\mathcal{J}$ with its set $\Lambda$ and a perfect subset $\lambda$ of $\Lambda$ with \begin{math}N_2>f(\lambda)\end{math}, let graph $H$ be the graph returned by $\mathcal{A}_1$.
    We consider a subset of vertices $W$ in graph $H$ that contains at least one vertex of path $a$.
    Subset $W$ inducing a subpath $H_W\approx L_1$ implies that the following Conditions~\ref{wextc},~\ref{w1ext} and~\ref{wa} are satisfied :
    \begin{enumerate}[(a)]
        \setcounter{enumi}{9}
        \item both extremities of $H_1$ are end subpaths of $\Lambda$, denoted $l$ and $l'$ hereafter;\label{wextc}
        \item an extremity of path $a$ is adjacent to the extremity of $H_1$ in $l\in\lambda$ ; \label{w1ext}
        \item we have $|l'|>|l|$ and \begin{math}N_2-f(\lambda)\geq|l'|-|l|\end{math}. \label{wa}
    \end{enumerate}
\end{lemma}
\begin{proof}
    \begin{figure}[ht]
    \centering
    \begin{subfigure}{0.45\textwidth}
        \centering
        \begin{adjustbox}{max width=\linewidth}
            \begin{tikzpicture}[node distance=1cm]
                \node[band, minimum width=1cm] (x1) {} ;
                \draw ([xshift=-.25cm,yshift=.1cm]x1.south west) node {$a$};
                \node[below of=x1] (x2) {$\ldots$} ;
                \node[pointNode, right of=x2] (x3) {} ;
                \draw ([xshift=-.15cm,yshift=-.2cm]x3.south east) node {$u$};
                \node[band, right of=x3, right=-.4cm, minimum width=3cm] (x4) {} ;
                \draw ([xshift=-1.5cm,yshift=-.2cm]x4.south east) node {$l$};
                \node[pointNode, right of=x4, right=.9cm] (x5) {} ;
                \draw ([xshift=-.15cm,yshift=-.2cm]x5.south east) node {$w$};
                \node[right of=x5] (x6) {$\ldots$} ;
                \node[voidNode, above of=x5] (x7) {} ;
                \node[right of=x7] (x8) {$\ldots$} ;
                
                \node[right of=x3, left=.15cm] (x9) {} ;
                \node[left of=x5, right=.15cm] (x10) {} ;
                    
                \draw[-] (x2) -- (x3) ;
                \draw[-] (x3) -- (x4) ;
                \draw[-] (x4) -- (x5) ;
                \draw[-] (x5) -- (x6) ;
                
                \draw[-] (x1) to[out=0, in=90] (x9) ;
                \draw[-] (x10) to[out=90, in=180] (x7) ;
                \draw[-] (x7) -- (x8) ;
            \end{tikzpicture}
        \end{adjustbox}
        \caption{Additional path $a$ has an extremity adjacent to one vertex in the middle of path $H_1$.}
        \label{fig:chainA-middle}
    \end{subfigure}\hfill
    \begin{subfigure}{0.48\textwidth}
        \centering
        \begin{adjustbox}{max width=\linewidth}
            \begin{tikzpicture}[node distance=1 cm]
                \node[band, minimum width=1cm] (x1) {} ;
                \draw ([xshift=.25cm,yshift=.1cm]x1.south east) node {$a$};
                \node[band, below of=x1, below=-.4cm, minimum width=2cm] (x2) {} ;
                \draw ([xshift=-1cm,yshift=-.2cm]x2.south east) node {$l$};
                \node[pointNode, right of=x2, right=0.3cm] (x3) {} ;
                \draw ([xshift=-.15cm,yshift=-.2cm]x3.south east) node {$u$};
                \node[right of=x3] (x4) {$\ldots$} ;
                \node[pointNode, right of=x4] (x5) {} ;
                \draw ([xshift=-.15cm,yshift=-.2cm]x5.south east) node {$w$};
                \node[band, right of=x5, right=-.5cm, minimum width=2cm] (x6) {} ;
                \draw ([xshift=-1cm,yshift=-.2cm]x6.south east) node {$l'$};
                
                \draw[-] (x1) to[out=180, in=180] (x2) ;
                \draw[-] (x2) -- (x3) ;
                \draw[-] (x3) -- (x4) ;
                \draw[-] (x4) -- (x5) ;
                \draw[-] (x5) -- (x6) ;
            \end{tikzpicture}
        \end{adjustbox}
        \caption{Additional path $a$ has an extremity adjacent to one extremity of path $H_1$.}
        \label{fig:chainA-border}
    \end{subfigure}
    \caption{Illustrations of an additional path $a$ position on graph $H$ returned by Algorithm $\mathcal{A}_1$.}
    \label{fig:l1-chainA}
\end{figure}
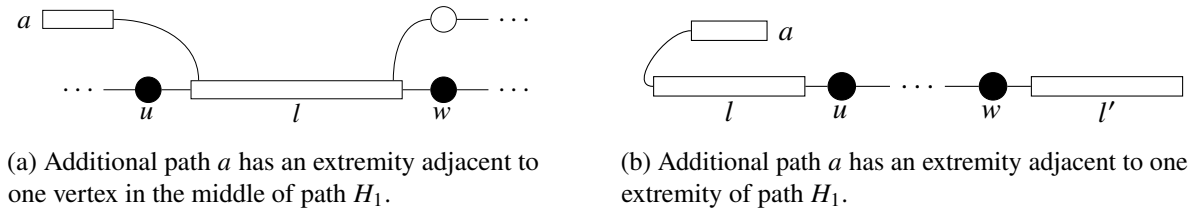
    By Lemma~\ref{lem:l1inH}, $W$ has $N_1$ vertices, induces a path, and contains all vertices of $H_1$ not coloured in $c$.
    Let $l\in\lambda$ be the maximal subpath coloured in $c$ to which $a$ is attached.
    By Lines~\ref{alg:line:alink1}--\ref{alg:line:alink3}, one extremity of path $a$ is adjacent to
    one outer extremity of $l$.

    The subpath $l$ must contain an extremity of $H_1$.
    Otherwise, let $u$ and $w$ be the closest vertices not coloured in $c$ of $H_1$ on the two sides of $l$, as in Figure~\ref{fig:chainA-middle}.
    As non-$c$ vertices, both belong to $W$.
    Since $W$ also contains a vertex of $a$, the attachment extremity of $l$ would have degree three in $H[W]$, contradicting the fact that $H[W]$ is a path.
    Moreover, the attachment extremity is therefore also an extremity of $H_1$, proving Condition~\ref{w1ext}.
    Without loss of generality, assume that $l$ contains the left extremity of $H_1$.

    By Lines~\ref{alg:line:y1link} and \ref{alg:line:chained}, any vertex adjacent to a vertex $v_i$ is its subpath's extremity, hence its third neighbour is not coloured in $c$.
    Therefore, Conditions~\ref{wch} and~\ref{wnonc} prevent $W$ from using any vertex $v_i$.
    Since $W$ contains a vertex of $a$ adjacent to one extremity of $L_1$, if the other extremity of $L_1$ is not coloured in $c$, it would then have to contain every vertex of $H_1$, contradicting $|W|=|H_1|$.
    Thus, both extremities of $H_1$ lie in end subpaths of $\Lambda$.
    Denote the other one by $l'$. This proves Condition~\ref{wextc}.

    Finally, since both extremities of $l$ are adjacent to one vertex in $a$ and to a vertex not coloured in $c$, respectively, the path condition forces $W$ to contain all vertices of $l$.
    Hence, the corresponding end subpath of $H[W]$ has length between $|l|+1$ and $|l|+|a|$.
    Since $H[W]\approx L_1$, it must correspond to the other end subpath $l'$ of $L_1$. 
    Consequently, \begin{math}|l|<|l'|\leq |l|+|a|\end{math}.
    Line~\ref{alg:line:test} attaches $a$ to the longer end subpath in $\lambda$; the inequality therefore implies $l'\notin\lambda$.
    As $|a|=N_2-f(\lambda)$, we conclude that \begin{math}N_2-f(\lambda)\geq |l'|-|l|\end{math}, proving Condition~\ref{wa}.
\end{proof}

The preceding lemmas characterise the possible obstructions of additional copies of $L_1$ distinct from the canonical copy $H_1$. 
We now show that each obstruction can be repaired by the following result on Algorithm~$\mathcal{A}_1$.

\begin{lemma}\label{lem:algo2}
    Given an instance in $\mathcal{J}$ with its set $\Lambda$ and a perfect subset $\lambda$ of $\Lambda$, let graph $H$ be the graph returned by $\mathcal{A}_1$. 
    If $H$ is not a solution graph, then there exists another perfect subset $\lambda_\star$ of $\Lambda$ with \begin{math}|\lambda|\geq|\lambda_\star|\end{math} and \begin{math}f(\lambda)\leq f(\lambda_\star)\leq N_2\end{math} for which Algorithm~$\mathcal{A}_1$ returns a solution graph with fewer vertices than $H$. 
\end{lemma}
\begin{proof}
    Since $\lambda$ is perfect, Lemmas~\ref{lem:algo1},~\ref{lem:l2-added}, and~\ref{lem:l2-ign} show that $H$ is connected and contains a unique induced copy $H_2\approx L_2$.
    Since all edges created after Line~\ref{alg:line:l1} are incident to a vertex of $H_2$, the canonical copy $H_1$ remains induced and $H_1\approx L_1$.
    Moreover, these edges belong to $H_2$, so $H$ is edge-covered by $\{L_1,L_2\}$.
    Hence, if $H$ is not a solution graph, it contains another induced copy $H[W]\approx L_1$, induced by a vertex set $W\neq V(H_1)$.
    Since every vertex outside $H_1$ is either a vertex $v_i$ or a vertex of $a$, we distinguish two cases.

    Suppose first that $W$ contains a vertex $v_i$.
    By Lemma~\ref{lem:l1wVi}, $|\lambda|=2$, $f(\lambda)=3<N_2$, and there is an $l\in\Lambda\setminus\lambda$ with $|l|=3$.
    Let $\lambda_\star=\{l\}$, which is perfect because $f(\lambda_\star)=f(\lambda)$, and $|\lambda_\star|<|\lambda|$.
    Let \begin{math}H_\star=\mathcal{A}_1(\lambda_\star)\end{math}.
    The direct properties above apply to $H_\star$, and Lemma~\ref{lem:Hsize} gives $|V(H_\star)|<|V(H)|$.
    Moreover, since the two extremities of $H_1$ belong to $\lambda$, $l$ contains no extremity of $H_1$.
    Thus, $\lambda_\star$ satisfies neither the exceptional conditions of Lemma~\ref{lem:l1wVi} nor those of Lemma~\ref{lem:l1wA}.
    Consequently, $H_\star$ has no additional copy of $L_1$ and is a solution graph.

    Suppose now that $W$ contains a vertex of $a$.
    Lemma~\ref{lem:l1wA} implies that there exist two end subpaths $l$ and $l'$ of $H_1$ such that $l\in\lambda$, $l'\notin\lambda$, and \begin{math}0<|l'|-|l|\leq N_2-f(\lambda)\end{math}.
    Define \begin{math}\lambda_\star=(\lambda\setminus\{l\})\cup\{l'\}\end{math}.
    It follows that 
    \begin{math}
        |\lambda_\star|=|\lambda|\text{ and }f(\lambda)<f(\lambda_\star)=f(\lambda)+|l'|-|l|\leq N_2
    \end{math}
    Since $f(\lambda_\star)\leq N_2$, Conditions~\ref{cond:twochains} and \ref{cond:ign} are automatic, and $\lambda_\star$ is non-empty; hence it is perfect.
    For $|V(H_\star)|<|V(H)|$, the direct properties above apply again, and Lemma~\ref{lem:Hsize} gives $|V(H_\star)|<|V(H)|$.
    Since $f(\lambda_\star)\leq N_2$, $\lambda_\star$ does not satisfy the exceptional configuration of Lemma~\ref{lem:l1wVi}; and the replacement also removes the exceptional attachment described by Lemma~\ref{lem:l1wA}.
    Consequently, $H_\star$ has no additional copy of $L_1$ and is a solution graph.

    As required, both cases lead to
    \begin{math}|V(H_\star)|<|V(H)|,~ |\lambda_\star|\leq|\lambda|~ \text{and}~ f(\lambda)\leq f(\lambda_\star)\leq N_2.\end{math}
\end{proof}

In the following, $H^*$ denotes an optimal graph of an instance in $\mathcal{J}$.
The following lemmas show that, when searching for an optimal graph, it suffices to consider the graphs produced by Algorithm~$\mathcal{A}_1$ from perfect subsets.
For a solution graph $H$, let $\lambda_H$ denote its associated subpath set.

\subsubsection{Normalising an Optimal Solution}

We now prove the converse structural direction. 
Starting from an arbitrary optimal solution, we show that one can retain optimality while obtaining a perfect subpath set. 
Together with the repair provided by Lemma~\ref{lem:algo2}, this establishes that the canonical construction captures the optimum.

The next two lemmas show that Conditions~\ref{cond:notempty} and \ref{cond:twochains} are satisfied by every solution.

\begin{lemma}\label{lem:sol-lambda-a}
    Given an instance in $\mathcal{J}$ and a solution graph $H$, its subpath set $\lambda_H$ satisfies Condition~\ref{cond:notempty}.
\end{lemma}
\begin{proof}
    Since the instance belongs to $\mathcal{J}$, $L_1$ does not contain $L_2$.
    Thus, $H_2$ has a vertex outside $H_1$.
    As $H$ is connected and every vertex belongs to $H_1\cup H_2$, some edge $(u,v)$ joins a vertex $u\in H_1$ to a vertex $v\in V_{2-1}$.
    This edge cannot be covered by $H_1$, and therefore belongs to $H_2$.
    Hence, $u\in V_{1+2}$, so $u$ belongs to a member of $\lambda_H$ and $\lambda_H\neq\emptyset$.
\end{proof}

\begin{lemma}\label{lem:sol-lambda-b}
    Given an instance in $\mathcal{J}$ and a solution graph $H$, its subpath set $\lambda_H$ satisfies Condition~\ref{cond:twochains}. 
\end{lemma}
\begin{proof}
    All cases are immediate except when $\lambda_H=\{l_1,l_2\}$ and $f(\lambda_H)>N_2$. We therefore consider this case.
    Since path $H_2$ uses vertices from two distinct maximal subpaths $l_1$ and $l_2$, there is a non-empty subpath $L$ whose endpoints are adjacent to vertices $u\in l_1$ and $w\in l_2$, respectively. 
    Thus, $u$ and $w$ are junction vertices of the subtree $T_\lambda$ of $H$.
    Each of the junction vertices $u$ and $w$ has three branches, two of which are contained in $l_1$ and $l_2$, respectively. 
    Denoting their orders by $p,q$ and $r,s$, we have $|l_1|=p+q+1$ and $|l_2|=r+s+1$. The configuration is as in Figure~\ref{fig:allInTwo}, with the path $L$ replacing $v_1$.

    Since $H_2$ is the longest subpath of colour $c$ in $H$, and it contains $N_2$ vertices including $u,L,w$, we have \begin{math}N_2=\max\{p,q\}+\max\{r,s\}+|L|+2\end{math}.
    
    Since $H_2$ is unique, we have $p\neq q$, $r\neq s$, hence 
    \begin{displaymath}
        N_2\geq \left\lceil\frac{|l_1|}{2}\right\rceil+ \left\lceil\frac{|l_2|}{2}\right\rceil+|L|+2.
    \end{displaymath}
    
    If Condition~\ref{cond:twochains} failed, then, since \begin{math}f(\lambda_H)=|l_1|+|l_2|+1\end{math}, integrality would give 
    \begin{displaymath}
        N_2\leq \left\lceil\frac{|l_1|}{2}\right\rceil+ \left\lceil\frac{|l_2|}{2}\right\rceil+2.
    \end{displaymath}
    
    These inequalities imply $|L|\leq0$, contradicting the fact that $L$ is non-empty.
    Thus, Condition~\ref{cond:twochains} holds.
\end{proof}

Since Condition~\ref{cond:ign} may fail for optimal solutions, we first establish two auxiliary results before proving that an optimal solution satisfying \ref{cond:ign} exists.

\begin{lemma}\label{lem:prfsubset}
    Given an instance in $\mathcal{J}$, a solution graph $H$ with $N_2<f(\lambda_H)$ and a perfect subset $\lambda_\star\subsetneq\lambda_H$.
    If $N_2-f(\lambda_\star)\leq|\lambda_H|-|\lambda_\star|$, then there is a solution graph with at most $|V(H)|$ vertices and a perfect subpath set.
\end{lemma}
\begin{proof}
    Let $H_\star=\mathcal{A}_1(\lambda_\star)$.
    By Lemma~\ref{lem:Hsize}, $|V(H_\star)|= N_1+|\lambda_\star|-1+ \max\{0,N_2-f(\lambda_\star)\}$.
    Since $\lambda_\star\subsetneq\lambda_H$ and $N_2-f(\lambda_\star)\leq|\lambda_H|-|\lambda_\star|$, we obtain $|\lambda_\star|<|\lambda_H|$ and $|\lambda_\star|+N_2-f(\lambda_\star)\leq|\lambda_H|$.
    Therefore, $|V(H_\star)|\leq N_1+|\lambda_H|-1$.
    
    Along $H_2$, two consecutive members of $\lambda_H$ require distinct vertices outside $H_1$ to join them; otherwise, $H_1$ would not be induced.
    Consequently, $|V(H)|\geq N_1+|\lambda_H|-1\geq |V(H_\star)|$.
    
    If $H_\star$ is not a solution graph, Lemma~\ref{lem:algo2} yields a strictly smaller solution graph having a perfect subpath set.
    In either case, there exists a solution graph with a perfect subpath set and at most $|V(H)|$ vertices.
\end{proof}

\begin{lemma}\label{lem:sol-max-ign}
    Given an instance in $\mathcal{J}$, a solution graph $H$ with $f(\lambda_H)>N_2$ and $|\lambda_H|>2$, let $l_a$ and $l_z$ be the subpaths of $\lambda_H$ containing the two extremities of $H_2$, respectively.
    The number of \emph{ignored} vertices in the subpaths of $\lambda_H$ is at least $f(\lambda_H)-N_2$ and at most 
    \begin{displaymath}
        \left\lfloor\frac{|l_a|}{2}\right\rfloor+\left\lfloor\frac{|l_z|}{2}\right\rfloor-2+\sum_{l_i\in\lambda_H\setminus\{l_a,l_z\}}\left(|l_i|-1\right).
    \end{displaymath}
    
\end{lemma}
\begin{proof}
    Since $H_1$ is induced, each pair of consecutive members of $\lambda_H$ along $H_2$ must be separated by at least one \emph{added} vertex.
    Thus, at least $|\lambda_H|-1$ of the $N_2$ vertices of $H_2$ are \emph{added} vertices.
    At most $N_2-(|\lambda_H|-1)$ vertices belonging to members of $\lambda_H$ can consequently be \emph{used}.
    The number of \emph{ignored} vertices is therefore at least 
    \begin{displaymath}
        \sum_{l_i\in\lambda_H}|l_i|-\bigl(N_2-(|\lambda_H|-1)\bigr)=f(\lambda_H)-N_2.
    \end{displaymath}

    For the upper bound, every member of $\lambda_H$ contains at least one \emph{used} vertex.
    A member other than $l_a$ and $l_z$ therefore contains at most $|l_i|-1$ \emph{ignored} vertices.
    Consider now $l_a$.
    If $H_2$ has no junction vertex in $l_a$, maximality of $H_2$ forces all vertices of $l_a$ to be \emph{used}.
    Otherwise, the junction separates a \emph{used} branch from an \emph{ignored} branch.
    Since branches of equal length would produce two induced copies of $L_2$, the \emph{used} branch must be strictly longer.
    Hence, the \emph{ignored} branch has at most $\lfloor |l_a|/2\rfloor-1$ vertices.
    The same argument applies to $l_z$.
    Summing the bounds, we obtain the claimed expression 
    \begin{displaymath}
        \left\lfloor\frac{|l_a|}{2}\right\rfloor+ 
        \left\lfloor\frac{|l_z|}{2}\right\rfloor-2+ 
        \sum_{l_i\in\lambda_H\setminus\{l_a,l_z\}}(|l_i|-1).
    \end{displaymath}
\end{proof}

We can now show that an optimal solution satisfying Condition~\ref{cond:ign} always exists.

\begin{lemma}\label{lem:opt-lambda-c}
    Given an instance in $\mathcal{J}$, there exists an optimal graph $H^*$ such that its subpath set $\lambda^*$ satisfies Condition~\ref{cond:ign}.
\end{lemma}
\begin{proof}
    Choose an optimal graph $H^*$ and write $\lambda^*=\lambda_{H^*}$.
    By Lemmas~\ref{lem:sol-lambda-a} and~\ref{lem:sol-lambda-b}, $\lambda^*$ satisfies Conditions~\ref{cond:notempty} and \ref{cond:twochains}.
    Suppose that it does not satisfy Condition~\ref{cond:ign}.
    Then $f(\lambda^*)>N_2$.

    Starting with $\mu_0=\lambda^*$, repeatedly remove a member $l_i\in\mu_i$ such that $|l_i|\leq f(\mu_i)-N_2$, and set $\mu_{i+1}=\mu_i\setminus\{l_i\}$.
    The process terminates at a non-empty subset $\lambda'$, because every singleton satisfies Condition~\ref{cond:ign}.
    By construction, $\lambda'$ satisfies that condition.

    The quantity controlled by a deletion is the deficit $d(\mu)=N_2-f(\mu)$.
    Indeed, 
    \begin{displaymath}
        f(\mu_{i+1})=f(\mu_i)-|l_i|-1 \qquad d(\mu_{i+1})=d(\mu_i)+|l_i|+1\leq1, \forall i\geq0.
    \end{displaymath}
    Since at least one deletion occurs, this yields
    \begin{equation}\label{eq:deletion-deficit}
       N_2-f(\lambda')\leq1 \leq |\lambda^*|-|\lambda'|.
    \end{equation}

    If $\lambda'$ satisfies Condition~\ref{cond:twochains}, it is perfect.
    Lemma~\ref{lem:prfsubset}, together with \eqref{eq:deletion-deficit}, 
    then gives a solution graph with at most $|V(H^*)|$ vertices and a perfect subpath set.
    By optimality, this solution is optimal, and the result follows.

    It remains to consider the only possible exceptional case: $\lambda'=\{l_1,l_2\}$ and Condition~\ref{cond:twochains} fails.
    Put $r=|\lambda^*|$; then $r>2$.
    Let $l_a$ and $l_z$ be the members of $\lambda^*$ containing the two extremities of $H_2$, and let $k$ be the number of \emph{ignored} vertices in the members of $\lambda^*$.
    By Lemma~\ref{lem:sol-max-ign},
    \begin{displaymath}
        f(\lambda^*)-N_2\leq k\leq
        \left\lfloor\frac{|l_a|}{2}\right\rfloor+
        \left\lfloor\frac{|l_z|}{2}\right\rfloor-2+
        \sum_{l_i\in\lambda^*\setminus\{l_a,l_z\}}(|l_i|-1).
    \end{displaymath}
    
    Rearranging gives
    \begin{equation}\label{eq:end-subpaths-lower}
       N_2\geq
       \left\lceil\frac{|l_a|}{2}\right\rceil+
       \left\lceil\frac{|l_z|}{2}\right\rceil+2r-1.
    \end{equation}
    Since Condition~\ref{cond:twochains} fails for $\{l_1,l_2\}$,
    \begin{equation}\label{eq:two-subpaths-upper}
       N_2\leq
       \left\lceil\frac{|l_1|}{2}\right\rceil+
       \left\lceil\frac{|l_2|}{2}\right\rceil+2.
    \end{equation}
    Equations~\eqref{eq:end-subpaths-lower} and \eqref{eq:two-subpaths-upper} show that $\{l_1,l_2\}\neq\{l_a,l_z\}$, since otherwise $2r-1\leq2$, contrary to $r>2$.

    Thus, one extremity of $H_2$ belongs to a removed member.
    Without loss of generality, let this member be $l=l_a$ and let $l_z=l_1$.
    From \eqref{eq:end-subpaths-lower} and $r\geq3$,
    \begin{equation}\label{eq:l-l1-lower}
       N_2\geq
       \left\lceil\frac{|l|}{2}\right\rceil+
       \left\lceil\frac{|l_1|}{2}\right\rceil+5.
    \end{equation}
    
    Suppose, for contradiction, that $|l|+|l_1|\leq N_2-r$ and $|l_2|\leq N_2-r$.
    Therefore, $|l_1|\leq N_2-r-|l|$, and \eqref{eq:two-subpaths-upper} would imply 
    \begin{displaymath}
        N_2 \leq \frac{N_2-r-|l|+1}{2} +\frac{N_2-r+1}{2}+2 =N_2-r-\frac{|l|}{2}+3,
    \end{displaymath}
    hence $r+|l|/2\leq3$, contradicting $r\geq3$ and $|l|\geq1$.
    Consequently, \begin{math}\max\{|l|+|l_1|,|l_2|\}\geq N_2-r+1\end{math}, and we consider the two cases separately.

    Suppose first that $|l|+|l_1|\geq N_2-r+1$, and set $\lambda_\star=\{l,l_1\}$.
    Equation~\eqref{eq:l-l1-lower} ensures Condition~\ref{cond:twochains} and Condition~\ref{cond:notempty} is immediate.
    Then 
    \begin{displaymath}
        N_2-f(\lambda_\star) =N_2-|l|-|l_1|-1\leq r-2.
    \end{displaymath}
    The instance belonging to $\mathcal{J}$ implies that $|l|+1\leq N_2$ and $|l_1|+1\leq N_2$.
    Since \begin{math}f(\lambda_\star)=|l|+|l_1|+1\end{math}, we obtain $|l|,|l_1|\geq f(\lambda_\star)-N_2$.
    If both inequalities are strict, $\lambda_\star$ also satisfies Condition~\ref{cond:ign} and is perfect.
    Lemma~\ref{lem:prfsubset} then produces an optimal solution with a perfect subpath set.
    If equality holds for one member, say $|l|=f(\lambda_\star)-N_2$, then $N_2=|l_1|+1$.
    The singleton $\{l_1\}$ is perfect, and Algorithm~$\mathcal{A}_1$ constructs a graph with $N_1+1$ vertices.
    Lemma~\ref{lem:algo2} guarantees a solution of no greater order.
    Since $r>2$, this order is strictly smaller than $|V(H^*)|$, contradicting the optimality of $H^*$.
    Thus, this equality case cannot occur.
    The case of equality for $l_1$ is symmetric.

    Finally, suppose that $|l_2|\geq N_2-r+1$.
    The singleton $\lambda_\star=\{l_2\}$ is perfect, and Lemma~\ref{lem:Hsize} gives 
    \begin{displaymath}
        |V(\mathcal{A}_1(\lambda_\star))|=N_1+\max\{0,N_2-|l_2|\}\leq N_1+r-1.
    \end{displaymath}
    Since $H^*$ contains at least $r-1$ added vertices, 
    \begin{displaymath}
        |V(H^*)|\geq N_1+r-1\geq|V(\mathcal{A}_1(\lambda_\star))|.
    \end{displaymath}
    If the constructed graph is not a solution, Lemma~\ref{lem:algo2} replaces it by a strictly smaller solution graph with a perfect subpath set, contradicting the optimality of $H^*$.
    Thus, in every case, there exists an optimal solution graph with a perfect subpath set.
\end{proof}

Lemmas~\ref{lem:sol-lambda-a},~\ref{lem:sol-lambda-b} and~\ref{lem:opt-lambda-c} imply the following corollary. 

\begin{corollary}\label{cor:opt-lambda}
    Given an instance in $\mathcal{J}$, there exists an optimal graph $H^*$ whose subpath set $\lambda_{H^*}$ is perfect. 
\end{corollary}

Given a set $\Lambda$, we define a function $q$ such that for every perfect subset $\lambda\subseteq\Lambda$, 
$q(\lambda)=|\lambda|-1+\max\{0,N_2-f(\lambda)\}$.

We now establish that Algorithm~$\mathcal{A}_1$ returns an optimal graph when applied to a perfect subset $\lambda$ that minimises the quantity $q(\lambda)$.

\begin{lemma}\label{lem:algo3}
    Given an instance in $\mathcal{J}$ and a perfect subset $\lambda$ of $\Lambda$ which minimises the quantity $q(\lambda)$, Algorithm~$\mathcal{A}_1$ returns an optimal graph. 
\end{lemma}
\begin{proof}
    Let $H=\mathcal A_1(\lambda)$.
    By Lemma~\ref{lem:Hsize}, $|V(H)|=N_1+q(\lambda)$.

    The direct properties of the construction show that $H$ is connected, contains the canonical induced copy of $L_1$, and contains a unique induced copy of $L_2$.
    If $H$ were not a solution graph, Lemma~\ref{lem:algo2} would provide a perfect subset $\lambda_\star$ for which $\mathcal A_1(\lambda_\star)$ is a solution graph strictly smaller than $H$.
    Lemma~\ref{lem:Hsize} would then give $q(\lambda_\star)<q(\lambda)$, contradicting the choice of $\lambda$.
    Hence, $H$ is a solution graph.

    It remains to prove optimality.
    By Corollary~\ref{cor:opt-lambda}, choose an optimal graph $H^*$ whose subpath set $\lambda^*$ is perfect.
    In the unique copy $H_2\approx L_2$, at least $|\lambda^*|-1$ vertices outside $H_1$ are needed to join consecutive members of $\lambda^*$;
    moreover, at most $\sum_{l\in\lambda^*}|l|$ vertices of $H_2$ lie in $H_1$.
    Therefore, $|V_{2-1}|\geq \max\left\{|\lambda^*|-1, N_2-\sum_{l\in\lambda^*}|l|\right\}=q(\lambda^*)$, and hence $|V(H^*)|\geq N_1+q(\lambda^*)$.

    On the other hand, Lemma~\ref{lem:Hsize} gives $|V(\mathcal A_1(\lambda^*))|=N_1+q(\lambda^*)$.
    If this graph were not a solution, Lemma~\ref{lem:algo2} would yield a solution graph of smaller order, contradicting the optimality of $H^*$.
    Thus, $|V(H^*)|=N_1+q(\lambda^*)$.
    Finally, $q(\lambda)\leq q(\lambda^*)$ by the choice of $\lambda$, whereas the optimality of $H^*$ gives $|V(H^*)|\leq|V(H)|$.
    Consequently, equality holds throughout, and $H$ is optimal.
\end{proof}

The preceding results yield the structural characterisation underlying the algorithm. 
For an instance in $\mathcal{J}$, let $\operatorname{OPT}(L_1,L_2)$ denote its minimum solution order. 
Then
\begin{equation}\label{eq:structural-characterisation}
    \operatorname{OPT}(L_1,L_2)
    =N_1+\min_{\substack{\lambda\subseteq\Lambda\setminus\lambda\text{ perfect}}}
    \max\{q(\lambda)\}.
\end{equation}
Indeed, Corollary~\ref{cor:opt-lambda} guarantees an optimal solution with a perfect subpath set, and Lemma~\ref{lem:algo3} shows that minimising the right-hand side produces an optimal solution. 
Equation~\eqref{eq:structural-characterisation} therefore separates the structural argument from the dynamic programme below.

\subsection{Optimising over Perfect Subsets}\label{sec:subpb}
For $\lambda\subseteq\Lambda$, write \begin{math}\sum\lambda=\sum_{l\in\lambda}|l|\end{math} for its total subpath length.

Consider an instance $(L_1,N_2,c)$ in $\mathcal{J}$, with the corresponding set $\Lambda$.
Since instances in $\mathcal{J}$ are \emph{without inclusion} (see Definition~\ref{def:ech}), every $l\in\Lambda$ satisfies $0<|l|<N_2$.

By Lemma~\ref{lem:algo3}, it remains to minimise the quantity $q(\lambda)$ over the perfect subsets $\lambda$ of $\Lambda$.
By Definition~\ref{def:funcf}, we have 
\begin{displaymath}
    q(\lambda)=\max\{|\lambda|-1,N_2-\sum\lambda\}.
\end{displaymath}

\begin{problem}[$\mathcal{P}$] \label{pb:P} 
    \item []\textbf{ Given } an integer $N_2>0$ and a set of subpaths $\Lambda$ where every $l\in\Lambda$ satisfies $0<|l|<N_2$, determine a subset $\lambda\subseteq\Lambda$ such that :
    \begin{displaymath}
        \lambda=\mathop{\arg\min}_{\substack{\lambda\subseteq\Lambda\setminus \lambda\text{ is perfect}}}
        \left(q(\lambda)=\max\left\{|\lambda|-1,\,N_2-\sum\lambda\right\}\right)
    \end{displaymath}
\end{problem}

Given a solution $\lambda$ of the problem $\mathcal{P}$, the following lemma establishes the range of $\sum\lambda$.

\begin{lemma}\label{lem:range}
    Given an instance of $\mathcal{P}$ and $\lambda$ its solution, $\sum\lambda< 2\times N_2$.
\end{lemma}
\begin{proof}
    Let $M=\max_{l\in\Lambda}|l|$.
    Suppose, for contradiction, that $\sum\lambda-N_2\geq M$.
    Since $\lambda$ is perfect, every $l\in\lambda$ satisfies $|l|>f(\lambda)-N_2 =\sum\lambda+|\lambda|-1-N_2 \geq M+|\lambda|-1\geq M$.
    This contradicts $|l|\leq M$.
    Hence, $\sum\lambda<N_2+M<2N_2$, because $M<N_2$.
\end{proof}

Observe first that no vertex can be \emph{ignored} when the subset contains a single subpath.
Moreover, any subset containing at least three subpaths automatically satisfies Conditions~\ref{cond:notempty} and~\ref{cond:twochains}. 
Consequently, the structural constraints defining a perfect subset differ significantly between the cases $1\leq|\lambda|\leq2$ and $3\leq|\lambda|\leq|\Lambda|$.
We therefore study these two cases separately.

\begin{lemma}\label{lem:lmbd12poly}
    Given an instance of $\mathcal{P}$, we can compute a perfect subset $\lambda\subseteq\Lambda$ of minimum quantity with at most two subpaths in $O(|\Lambda|^2)$ time.
\end{lemma}
\begin{proof}
    Enumerate the $|\Lambda|$ singletons and the $\binom{|\Lambda|}{2}$ two-element subsets.
    For each candidate, Conditions~\ref{cond:notempty}--\ref{cond:ign} and the quantity $\max\{|\lambda|-1,N_2-\sum\lambda\}$ can be evaluated in constant time.
    Retaining a feasible candidate of minimum quantity therefore takes $O(|\Lambda|^2)$ time.
\end{proof}

It remains to handle the case with $3\leq|\lambda|\leq|\Lambda|$.
We introduce the following auxiliary problem, which asks whether there exists a perfect subset with prescribed cardinality $K$ and total length $\Sigma$.

\begin{problem}[$\mathcal{P}_{K,\Sigma}$] \label{pb:PKS}
    \item []\textbf{ Given :} an instance $(N_2,\Lambda)$ of $\mathcal{P}$ and two integers $K\geq3$ and $\Sigma\in[K,2\times N_2-1]$. 
    \item []\textbf{ Question :} Does there exist a subset $\lambda\subseteq\Lambda$ such that $|\lambda|=K$, $\sum\lambda=\Sigma$ and every $l\in\lambda$ satisfies $|l|>\Sigma+K-1-N_2$?
\end{problem}

We denote by $\lambda_{K,\Sigma}$ a solution of $\mathcal{P}_{K,\Sigma}$, whenever one exists.

\begin{lemma}\label{lem:PKSperfect}
    Given a solution $\lambda_{K,\Sigma}$ of the problem $\mathcal{P}_{K,\Sigma}$, then $\lambda_{K,\Sigma}$ is a perfect subset whose quantity equals $\max\{K-1,N_2-\Sigma\}$.
\end{lemma}
\begin{proof}
    Since $|\lambda_{K,\Sigma}|=K\geq3$, Conditions~\ref{cond:notempty} and~\ref{cond:twochains} hold automatically.
    Moreover, $f(\lambda_{K,\Sigma})=\Sigma+K-1$, and the defining inequality of $\mathcal P_{K,\Sigma}$ is exactly Condition~\ref{cond:ign}.
    Thus, $\lambda_{K,\Sigma}$ is perfect.
    Substituting its cardinality and total length in the definition of the quantity gives $\max\{K-1,N_2-\Sigma\}$.
\end{proof}

Given a set $\Lambda$ and values $N_2$, $K$ and $\Sigma$, let $\Lambda_{K,\Sigma}$ be the subset of $\Lambda$ such that 
\begin{displaymath}
    \Lambda_{K,\Sigma}=\left\{l\in\Lambda\;\middle|\;|l|>\Sigma+K-1-N_2\right\}.
\end{displaymath}
By definition, every solution of $\mathcal{P}_{K,\Sigma}$ is a subset of $\Lambda_{K,\Sigma}$.
Once $\Lambda$ is sorted by decreasing subpath length, every $\Lambda_{K,\Sigma}$ is a prefix of this ordering. 
The sorting takes $O(|\Lambda|\log|\Lambda|)$ time, and the endpoint of the relevant prefix can subsequently be found from the threshold $\Sigma+K-1-N_2$ by binary search. 
This cost is dominated by the dynamic programmes below.

We will solve $\mathcal{P}_{K,\Sigma}$ through the following auxiliary problem.

\begin{problem}[$\mathcal{Q}_{i,k,\sigma}$] \label{pb:Qiks}
    \item []\textbf{ Given :} set $\Lambda_{K,\Sigma}=\{l_1,\ldots,l_{|\Lambda_{K,\Sigma}|}\}$ and three integers $i$, $k$ and $\sigma$ satisfying $0\leq i\leq|\Lambda_{K,\Sigma}|$, $0\leq k\leq K$ and $0\leq\sigma\leq\Sigma$. 
    \item []\textbf{ Question :} Does there exist a subset $\lambda\subseteq\{l_1,\ldots,l_i\}$ such that $|\lambda|=k$ and $\sum\lambda=\sigma$?
\end{problem}

Whenever $\mathcal{Q}_{i,k,\sigma}$ admits a solution, we denote one of its solutions by $\lambda_{i,k,\sigma}$.

\begin{lemma}\label{lem:PKSpoly}
    Given an instance of $\mathcal{P}_{K,\Sigma}$ and set $\Lambda_{K,\Sigma}$, we can determine whether it admits a solution and construct one, whenever it exists, in $O(|\Lambda|\cdot K \cdot\Sigma)$ time.
\end{lemma}
\begin{proof}
    Put $m=|\Lambda_{K,\Sigma}|$.
    By definition of $\Lambda_{K,\Sigma}$, the problem $\mathcal P_{K,\Sigma}$ has a solution if and only if $\mathcal Q_{m,K,\Sigma}$ does.

    For $0\leq i\leq m$, $0\leq k\leq K$, and $0\leq\sigma\leq\Sigma$, let $Q[i,k,\sigma]$ denote the truth value of $\mathcal Q_{i,k,\sigma}$.
    The boundary conditions and recurrence are
    \begin{displaymath}
        Q[i,k,\sigma] = 
        \left\{
        \begin{array}{ll}
            \texttt{False}        & \text{if } i<k \text{ or } \sigma<k,        \\[0.2cm]
            \texttt{False}        & \text{if } k=0 \text{ and } \sigma>0,       \\[0.2cm]
            \texttt{False}        & \text{if } \sigma=0 \text{ and } k>0,       \\[0.2cm]
            \texttt{True}         & \text{if } k=0 \text{ and } \sigma=0,       \\[0.2cm]
            Q[i-1,k,\sigma]       & \text{if } \sigma<|l_i|,                    \\[0.2cm]
            Q[i-1,k,\sigma]\lor Q[i-1,k-1,\sigma-|l_i|] & \text{otherwise.}
        \end{array}
        \right.
    \end{displaymath}
    Indeed, a feasible subset either excludes $l_i$, or contains $l_i$ and leaves a subset of cardinality $k-1$ and total length $\sigma-|l_i|$ among the first $i-1$ elements.
    The converse follows by the same decomposition, which proves the recurrence.

    There are $O(mK\Sigma)$ states, each evaluated in constant time.
    Thus, the table is computed in $O(mK\Sigma)=O(|\Lambda|K\Sigma)$ time.
    Storing one valid predecessor for every true state reconstructs a solution, when one exists, within
    the same bound.
\end{proof}

To solve the problem $\mathcal{P}$, we combine the search for a perfect subset of minimum quantity and at most two subpaths with the search for a perfect subset of minimum quantity and at least three subpaths.

\begin{lemma}\label{lem:subset-poly}
    Given an instance of $\mathcal{P}$, we can solve it in $O(|\Lambda|^3\cdot{N_2}^2)$ time.
\end{lemma}
\begin{proof}
    First compute an optimal perfect subset of cardinality at most two by Lemma~\ref{lem:lmbd12poly}.
    For cardinality at least three, enumerate $3\leq K\leq|\Lambda|$ and $K\leq\Sigma\leq2N_2-1$.
    
    These ranges are exhaustive: positivity of the subpath lengths gives $\Sigma\geq K$, and Lemma~\ref{lem:range} gives $\Sigma<2N_2$ for an optimal perfect subset.
    For each pair $(K,\Sigma)$, solve $\mathcal P_{K,\Sigma}$ by Lemma~\ref{lem:PKSpoly}; whenever it is feasible, Lemma~\ref{lem:PKSperfect} gives a perfect subset of quantity $\max\{K-1,N_2-\Sigma\}$.
    Comparing all feasible candidates with the best candidate of cardinality at most two therefore solves $\mathcal P$.

    Each pair is processed in $O(|\Lambda|K\Sigma)$ time, so the total running time is
    \begin{displaymath}
        O\left(
        \sum\limits_{K=3}^{|\Lambda|}
        \sum\limits_{\Sigma=K}^{2\times N_2-1}
        (|\Lambda|\cdot K\cdot\Sigma)
        \right).
    \end{displaymath}
    Using $\sum_{K=3}^{|\Lambda|}K=O(|\Lambda|^2)$ and $\sum_{\Sigma=K}^{2N_2-1}\Sigma=O(N_2^2)$ yields $O(|\Lambda|^3N_2^2)$.
    The initial sorting and the $O(|\Lambda|^2)$ enumeration from Lemma~\ref{lem:lmbd12poly} are dominated by this bound.
\end{proof}

We now conclude the proof of Theorem~\ref{the:classP}.

\begin{proof}{\emph{of Theorem}~\ref{the:classP}}
    Let an instance in $\mathcal I$ be given, and put $N_i=|V(L_i)|$.
    Let $c$ be the colour of the monochromatic path $L_2$ and let $\Lambda$ be the set of maximal subpaths coloured in $c$ of $L_1$.
    A linear scan of $L_1$ constructs $\Lambda$ and simultaneously tests whether some member has length at least $N_2$, which is equivalent to $L_2$ being included in $L_1$.
    Together with the test $\alpha>\max\{N_1,N_2\}$, this decides in linear time whether the instance belongs to $\mathcal J$.
    Instances outside $\mathcal J$ are handled in polynomial time by Lemma~\ref{lem:IexJ}.

    Suppose therefore that the instance belongs to $\mathcal J$.
    Then $0<|l|<N_2$ for every $l\in\Lambda$, so $(N_2,\Lambda)$ is an instance of $\mathcal P$.
    Since $|\Lambda|\leq N_1$, Lemma~\ref{lem:subset-poly} computes a perfect subset $\lambda^*$ of minimum quantity in 
    \begin{displaymath}
        O(|\Lambda|^3N_2^2)=O(N_1^3N_2^2).
    \end{displaymath}
    Lemmas~\ref{lem:Hsize} and \ref{lem:algo3} then show that $\mathcal A_1(\lambda^*)$ is an optimal graph of order 
    \begin{displaymath}
        N_1+\max\left\{|\lambda^*|-1, N_2-\sum\lambda^*\right\}.
    \end{displaymath}
    Hence, the original instance is feasible precisely when this value is at most $\alpha$.

    Finally, $\mathcal A_1$ constructs the optimal graph in $O(N_1|\lambda^*|)=O(N_1^2)$ time, which is dominated by the preceding optimisation.
    Both the decision and constructive versions are therefore polynomial-time solvable.
\end{proof}

\section{Conclusion}
This paper establishes several complexity results for SGSC. 
We first prove that SGSC is $\Sigma_2^P$-complete for instances with three colours, a sample $\mathcal{F}$ consisting of two graphs, and a unitary fingerprint $P$. 
We then show that the problem is $NP$-hard when the sample consists of only two paths with an arbitrary number of colours. 
On the positive side, we give a polynomial-time algorithm for two-path instances when one of the paths is monochromatic.

These results raise the question of where exactly the boundary between polynomial-time and NP-complete instances lies. 
In particular, is this boundary determined by the diversity of colours, or rather by how colours are mixed along the two paths? 
One possible tractable class would consist of paths that are concatenations of monochromatic subpaths, with no colour appearing in more than one subpath. 
Determining whether such a structural condition is sufficient for polynomial-time solvability, and more generally identifying the precise boundary between the two regimes, remains open.

Further directions include extending the polynomial-time approach to monochromatic trees or to $k$-coloured paths, for fixed or unrestricted $k$. 
Motivated by the molecular-graph applications discussed in the introduction, another natural extension is to minimise the size of a graph containing a polynomial number of copies of fixed small graphs. 
This variant lies in $NP$, since testing and counting occurrences of fixed-size graphs can be performed in polynomial time. 

Finally, these questions motivate a parameterised study of SGSC with parameters such as $|\mathcal{F}|$, $\sum_i P_i$, and $\max_{G_i\in\mathcal{F}} |G_i|$.

\vspace{1em}
\noindent\rule{\textwidth}{0.4pt}
\vspace{0.5em}

\noindent\textbf{AI Disclosure}

We used GPT-6 to enhance readability and check spelling. No AI tools have contributed to the mathematical content of the manuscript.

\nocite{*}
\bibliographystyle{abbrvnat}
\bibliography{biblio}
\label{sec:biblio}

\end{document}